%% file: root.tex
\documentclass[journal,twoside,web]{ieeecolor}
\usepackage{generic}
\usepackage{textcomp}
\usepackage{graphicx}
\usepackage{amsmath}
\usepackage{optidef}
 
\usepackage{amsthm}
\usepackage{amssymb}
\usepackage{amsfonts}       
\usepackage{multicol}
\usepackage{cuted}
\usepackage{multirow}
\usepackage{caption, subcaption}
\newtheorem{theorem}{Theorem}
\newtheorem{lemma}[theorem]{Lemma}
\newtheorem{proposition}[theorem]{Proposition}

\newtheorem{problem}{Problem}
\newtheorem{assumption}{Assumption}

\newtheorem{corollary}{Corollary}[theorem]
\usepackage{xcolor}
\usepackage{cite}
\usepackage[colorlinks=true,allcolors=steelblue]{hyperref}
\usepackage{comment}
\usepackage{algorithm}
\usepackage{algorithmic}
\usepackage{scalerel}
\usepackage{bm}
\usepackage{tikz}
\usepackage{tikz-3dplot}
\usepackage{pgfplots}
\usepackage{pst-plot}
\usetikzlibrary{positioning}
\usetikzlibrary{arrows.meta}
\pgfplotsset{compat=1.11}
\usetikzlibrary{shapes,arrows,angles,patterns,calc, intersections, quotes}
\definecolor{steelblue}{RGB}{70,130,180}

\input{commands.tex}

\def\BibTeX{{\rm B\kern-.05em{\sc i\kern-.025em b}\kern-.08em
    T\kern-.1667em\lower.7ex\hbox{E}\kern-.125emX}}
\begin{document}
\title{Probabilistic Robustness in the Gap Metric}
\author{Venkatraman Renganathan, \IEEEmembership{Member, IEEE}
\thanks{V. Renganathan is with the Faculty of Engineering \& Applied Sciences, Cranfield University, UK. Email: v.renganathan@cranfield.ac.uk.}
}

\maketitle

\begin{abstract}

Uncertainties influencing the dynamical systems pose a significant challenge in estimating the achievable performance of a controller aiming to control such uncertain systems. When the uncertainties are of stochastic nature, obtaining hard guarantees for the robustness of a controller aiming to hedge against the uncertainty is not possible. This issue set the platform for the development of probabilistic robust control approaches. In this work, we utilise the gap metric between the known nominal model and the unknown perturbed model of the uncertain system as a tool to gauge the robustness of a controller and formulate the gap as a random variable in the setting with stochastic uncertainties. The main results of this paper include giving a probabilistic bound on the gap exceeding a known threshold, followed by bounds on the expected gap value and probabilistic robust stability and performance guarantees in terms of the gap metric. We also provide a probabilistic controller performance certification under gap uncertainty and probabilistic guarantee on the achievable $\mathcal{H}_{\infty}$ robustness. Numerical simulations are provided to demonstrate the proposed approach.
\end{abstract}



\section{Introduction}\label{sec_intro}
\IEEEPARstart{R}{obust} control stands to be one of the most mature control methodologies to be ever developed mainly due to the strong guarantees that comes with it (interested readers are referred to \cite{ZhouDoyleBook, GreenLimebeerBook} and the references therein). Vinnicombe in \cite{vinnicombe2001uncertainty} describes robust control approaches as the ones where we try to come up with a control input for a system using what we know about the system so that the control input renders the system insensitive to what we do not know about the system. To illustrate that thought, consider the Figure \ref{fig:robustcontroldiagram}, where $\bar{\Sigma}$ denotes the known nominal model of the system and $\tilde{\Sigma}$ as the (true and possibly unknown) perturbed model of the system. The perturbed model $\tilde{\Sigma}$ is obtained by combining $\bar{\Sigma}$ and the uncertainty $\Delta$ that encapsulates what we do not know about the system as per Vinnicombe's description. Then, informally speaking, one can describe
\begin{align}
\label{eqn_informal_set_of_of_all_plants}
\tilde{\Sigma}
=
\left\{
\mathrm{combination}(\bar{\Sigma}, \Delta) \mid \text{ where }  
\Delta \in \mathbf{\Delta}
\right\},    
\end{align}
where, $\mathbf{\Delta}$ denotes the set of possible uncertainties and it is allowed to be structured, unstructured, parametric, static, dynamic, time invariant and even time-varying in nature. 

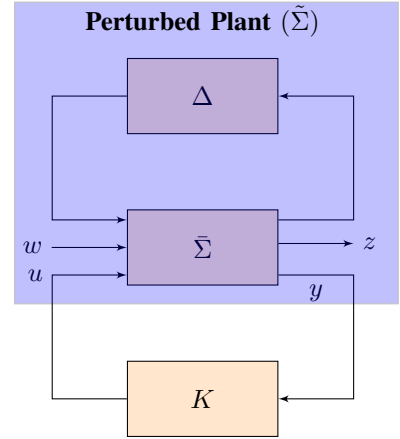
\begin{figure}
\centering
\input{robustcontroldiagram}   
\caption{A perturbed plant $\tilde{\Sigma}$ given by combination of a nominal plant $\bar{\Sigma}$ \& the uncertainty $\Delta$ is shown here. In robust control, the controller $K$ aims to minimize the transfer function $\mathbf{T}_{zw}$ between $w$ (external input) and $z$ (regulated output) for all combinations of $\bar{\Sigma}$ and $\Delta$ resulting in $\tilde{\Sigma}$.}
\label{fig:robustcontroldiagram}
\end{figure}

The description of perturbed model of the system $\tilde{\Sigma}$ using \eqref{eqn_informal_set_of_of_all_plants} naturally led several people to conceive the concept of measuring the distance between systems (specifically between Linear Time Invariant (LTI) systems which is of interest to us in this paper). Vidyasagar in \cite{vidyasagar1984graph} proposed the graph metric, followed by gap metric which was proposed in \cite{el1985gap}. Since the graph metric was difficult to compute, Georgiou in \cite{georgiou1988computation} came up with an elegant formula to compute the gap metric and subsequently the optimal robustness in the gap metric was established in \cite{tryhon_smith_tac_1990}. Meanwhile, Glover in \cite{glover1984all} had extended the famous small-gain theorem in \cite{zamesSmallGain} to handle perturbation in $\mathcal{L}_{\infty}$ rather than only in $\mathcal{H}_{\infty}$. Vinnicombe leveraged this development and proposed a new metric called the $\nu$-gap in \cite{vinnicombe1993frequency} by building upon the gap metric. Authors in \cite{lanzon2009distance} proposed a generalised distance measure for LTI systems by taking into account the information about several uncertainty structures. While all these metrics can be used for proposing probabilistic robustness, we shall be employing only the gap metric in this paper. Though gap metric has also been extensively studied for linear time varying (LTV) systems setting too in \cite{cantoni2000gap, djouadi2015robustness}, we shall restrict our study in this manuscript to LTI systems and leave the extension to LTV systems for future work. \\

Probabilistic Robust Control (PRC) approaches have been explored before and the interested readers are referred to \cite{stengel1991technical, khargonekar1996randomized, calafiore2007probabilistic, tempo2013randomized} and the references therein. In connection with the scenario-based approaches for robust control design, see \cite{calafiore2006scenario} and the references therein for more details. Having said that, the approach proposed in this paper involves metric between dynamical systems and casting them as random quantities, while the existing PRC approaches define robustness guarantees in terms of the volume of destabilizing perturbations or equivalently probability of violation of robust stability and performance conditions. Though our approach presumably aims to solve the same problem, the underlying methodology is certainly different from the existing PRC approaches. In this paper, we shall be using the gap metric for gauging the robustness and use the tools from high dimensional statistics for giving probabilistic guarantees. It was already noted by authors in \cite{safanov_gap} that the gap metric may be not suitable for evaluating the closeness of systems having uncertain poles and zeros on or near the imaginary axis resulting in the difficulty of the stability of the perturbed plant with a stabilizing controller designed for the nominal plant not being guaranteed through the existing stability theorems involving the gap metric. Our developments using probabilistic gap based problem formulation in this manuscript will aim to investigate along with their findings in the sense of probabilistic guarantees.

\subsection{Contributions}
This work is an extension of \cite{venkat_cdc_2024} where probabilistic robustness in terms of the $\nu$-gap metric in the frequency domain was initiated. In this work, a similar extension is sought albeit in the time domain using the gap metric. We believe our new perspective on PRC theory using gap between dynamical systems will further strengthen the existing theory on PRC and will open new doors for further exploration. The main contributions of this article are as follows: 
\begin{enumerate}
    \item When a random parameter affects a linear system, we study the satisfaction of the B\'ezout Identity governing the normalised co-prime factors of the uncertain system transfer function under that parameter uncertainty (See Lemma \ref{lemma_Bezout_Gaussian}) and subsequently we investigate the probabilistic guarantee associated with the randomness in the coprime factor uncertainty in Theorem \ref{theorem_coprime_probabilistic_guarantee}.
    \item When the perturbed model $\tilde{\Sigma}$ is not known exactly due to the random parameters, we formulate the associated gap metric between the known nominal model $\bar{\Sigma}$ and $\tilde{\Sigma}$ as a random variable and study the 
    probability of the random gap exceeding a known threshold in Theorem \ref{theorem_gap_lipschitz} and its corollaries and give bounds on its expected value in Lemma \ref{lemma_expected_gap} and its associated corollary.
    \item We formulate the randomness in the perturbed model $\tilde{\Sigma}$ by resulting it from the stochastic parametric uncertainty for the LTI system and we discuss probabilistic robust stability result in Theorem \ref{theorem_probabilistic_stability} and we provide probabilistic closed loop deviation (in the $\mathcal{H}_{\infty}$ norm) guarantees due to the action of a nominal controller aiming to control the perturbed model (See Theorem \ref{theorem_Q_difference}). 
    \item We give probabilistic $\mathcal{H}_{\infty}$ performance bound for the perturbed system under uncertainty in Theorem \ref{theorem_probability_hinfty_performance_satisfaction} and give probabilistic guarantee of meeting a desired $\mathcal{H}_{\infty}$ performance in Corollary \ref{corollary_Hinf_performance_probability_lower_bound}. We also give a rather conservative bound on the expected value of the $\mathcal{H}_{\infty}$ norm of the random perturbed system in Theorem \ref{theorem_expected_Hinf_norm_Tzw}.
    \item We connect the random gap metric problem formulation with the scenario based robustness approach to give probabilistic robust stability result in terms of the gap metric based performance measure in Theorem \ref{theorem_scenario_gap_robustness}.
\end{enumerate}

Numerical simulations are given at many places along the manuscript to demonstrate the idea proposed in this paper. The Matlab codes to reproduce the results provided in this manuscript are available at \url{https://github.com/venkatramanrenganathan/probust-gap}.


\subsection{Notations} 
The cardinality of the set $A$ is denoted by $\left | A \right \vert$. Given two sets $A, B$ such that $A \subset B$, the notation $A^{c}$ denotes the complement of set $A$ in $B$ meaning that $A^{c} := \{x \in B \mid x \notin A\} \subset B$. The set of real numbers, integers and the natural numbers are denoted by $\bbr, \bbz, \bbn$ respectively and the subset of real numbers greater than a given constant say $a \in \bbr$ is denoted by $\bbr_{> a}$. 
The subset of natural numbers between two constants $a, b \in \bbn$ with $a < b$ is denoted by $[a:b]$. 
For a matrix $A \in \bbr^{n \times n}$, we denote its transpose and its trace by $A^{\top}$ and $\mathbf{Tr}(A)$ respectively. An identity matrix of dimension $n$ is denoted by $I_{n}$.
We denote by $\mathbb{S}^{n}$ the set of symmetric matrices in $\bbr^{n \times n}$. 
For $A \in \mathbb{S}^{n}$, we denote by $A \succ 0 (A \succeq 0)$ to mean that $A$ is positive definite (positive semi-definite). 
Given $x \in \bbr^{n}$, the notation $\norm{x}$ denotes the $\mathcal{L}_{2}$ norm and is given by $\sqrt{x^{\top} x}$. For brevity of notation, multi-variate functions like $g(x(t), y(t))$ shall be abbreviated as $g(t; x, y)$. For a subspace $S \subset \mathbb{R}^{n}$, its orthogonal complement is denoted by $S^{\perp}$. The angle between two vectors $x, y \in \mathbb{R}^{n}$ and the angle between two subspaces $M, N \subset \mathbb{R}^{n}$ are denoted by $\angle(x,y)$ and $\angle(M,N)$ respectively. \\

The probability space is defined using a triplet $(\Omega, \mathcal{F}, \mathbb{P})$, where $\Omega, \mathcal{F}$, and $\mathbb{P}$ denote the sample space, event space, and the probability function, respectively, with $\mathbb{P}: \Omega \rightarrow [0,1]$. A real random vector $x \in \mathbb{R}^{n}$ following a probability density function $\mathbf{f_{x}}$ is denoted by $x \sim \mathbf{f_{x}}$ and its expectation is denoted by $\mathbb{E}[x]$. A zero-mean random vector $x \in \mathbb{R}^{n}$ following a Gaussian distribution with covariance $I_{n} \succ 0$ is denoted by $x \sim \mathcal{N}(0, I_{n})$. Similarly, a random variable $x$ following a Chi-squared distribution with parameter $p > 0$ is denoted by $x \sim \chi^{2}_{p}$. A random variable $x \in \mathbb{R}$ is said to be sub-Gaussian with parameter $\sigma > 0$ if $\forall \lambda \in \mathbb{R}$, it satisfies $\mathbb{E}[e^{\lambda (X - \mathbb{E}[X])}] \leq e^{\frac{\lambda^2 \sigma^2}{2}}$
or equivalently if $\forall \epsilon > 0$, it satisfies $\mathbb{P}(x - \mathbb{E}[x] \geq \epsilon) \leq \exp\left( -\frac{\epsilon^2}{2 \sigma^2} \right)$. The covariance of a random vector $x \in \mathbb{R}^{n}$ is denoted by $\mathbf{Cov}(x)$. \\

Let $\mathbf{R}(s)$ denote the set of rational functions in $s \in \mathbb{C}$ with real coefficients. We use $\mathcal{P}(s) \subset \mathbf{R}(s)$ to denote the set of proper rational functions whose poles are in the open left half-plane. Conceptually speaking, $\mathcal{P}(s)$ denotes the set of all finite-dimensional stable systems. Let us denote the set of matrices with elements in $\mathbf{R}(s)$ as $\mathrm{mat}(\mathbf{R}(s))$ and similarly, let us denote the set of matrices with elements in $\mathcal{P}(s)$ as $\mathrm{mat}(\mathcal{P}(s))$. Let $\mathcal{L}_{2}$ denote the space of all signals, or vectors of signals with bounded energy. In the frequency domain, the space $\mathcal{L}_{2}$ can be decomposed into $\mathcal{H}_{2}$ and $\mathcal{H}^{\perp}_{2}$, where $\mathcal{H}_{2}$ denotes the space of the Fourier transforms of signals defined for positive time and zero for negative time and $\mathcal{H}^{\perp}_{2}$ denotes the space of the Fourier transforms of signals defined for negative time and zero for positive time. The Hardy space consisting of transfer functions of stable LTI continuous time systems is denoted by $\mathcal{H}_{\infty}$ and is equipped with the $\mathcal{H}_{\infty}$ norm $\forall P(s) \in \mathrm{mat}(\mathcal{P}(s))$ given by
\begin{align} \label{eqn_norm_G_of_s}
    \norm{P(s)}_{\mathcal{H}_{\infty}} := \sup_{\omega \in [0, \infty)} \bar{\sigma}(P(j\omega)),
\end{align}
with $\bar{\sigma}(G(j\omega))$ denoting the maximum singular value. 
It happens that, $\mathcal{H}_{\infty}$ norm equals the induced norm. That is,
\begin{align}\label{eqn_h_infty_norm_def}
    \norm{P}_{\mathcal{H}_{\infty}} 
    := \sup_{u \in \mathcal{H}_{2}, u \neq 0} \frac{\norm{Pu}_{\mathcal{H}_{2}}}{\norm{u}_{\mathcal{H}_{2}}}.
\end{align}
The notation $\mathbf{R}\mathcal{H}_{\infty}$ denotes the set of all stable rational transfer functions.

\subsection{Preliminaries About Gap Metric}
Given $P(s) \in \mathrm{mat}(\mathbf{R}(s))$, we define its associated multiplication operator on $\mathcal{H}_{2}$ (which we identify with $\mathcal{L}_{2}[0, \infty)$ in the time domain) as
\begin{align}
\mathbf{M}_{P} 
:
\mathcal{H}_{2} \rightarrow \mathcal{H}_{2}, u \rightarrow Pu
\end{align}
and the domain of the operator as
\begin{align}
\mathcal{D}(\mathbf{M}_{P})
:=
\{
u \in \mathcal{H}_{2} \mid Pu \in \mathcal{H}_{2}
\}.
\end{align}
The graph of the operator $\mathbf{M}_{P}$ is the set of all possible bounded input/output pairs and is denoted by  
\begin{align}
\label{eqn_graph_of_operator}
\mathrm{graph}(\mathbf{M}_{P}) 
:=
\left\{ 
\begin{bmatrix}
Pu \\
u
\end{bmatrix}
:
u \in \mathcal{D}(\mathbf{M}_{P})
\right\}.
\end{align}
Every matrix $P(s) \in \mathrm{mat}(\mathbf{R}(s))$ has both a Right Co-prime Factorisation (RCF) as well as a Left Co-prime Factorisation (LCF) over the ring $\mathcal{P}(s)$. That is, $\forall P(s) \in \mathrm{mat}(\mathbf{R}(s))$, there exist $N, D, \Tilde{N}, \Tilde{D}, X, Y, \Tilde{X}, \Tilde{Y} \in \mathrm{mat}(\mathcal{P}(s))$ such that 
\begin{align}
\label{eqn_coprime_representation}
P
= 
N D^{-1} 
= 
\Tilde{D}^{-1}\Tilde{N},    
\end{align}
and the following \emph{Bezout's identity} holds for all $s \in \mathbb{C}_{\geq 0}$, 
\begin{align*}
X(s)N(s) + Y(s)D(s) = \Tilde{N}(s) \Tilde{X}(s) + \Tilde{D}(s) \Tilde{Y}(s) = I.
\end{align*} 
Further, the RCF is said to be \emph{normalized} if in addition it satisfies $N^{\star} N + D^{\star} D = I$. Analogous LCF results are available and are omitted here for the reason of being not used in this paper. Given $P(s) \in \mathrm{mat}(\mathbf{R}(s))$, its $\mathcal{H}_2$-graph is defined as
\begin{subequations}
\label{eqn_graph_p}
\begin{align}
    \mathcal{G}_{P} 
    &:= 
    \left\{ (u,y) \mid y = Pu \right\} \subseteq \mathcal{H}_{2} \times \mathcal{H}_{2} \\
    &=
    \underbrace{
    \begin{bmatrix}
    D \\ N    
    \end{bmatrix}}_{=: G}
    \mathcal{H}_{2} = \mathrm{Range}(G),
\end{align}
\end{subequations}
where the operator $G$ (henceforth referred to as the graph symbol) is unitary meaning that $G^{\star} G = I$. Note that $\mathcal{G}_{P}$ is a closed subspace of $\mathcal{H}_{2} \times \mathcal{H}_{2}$. The orthogonal projection onto $\mathcal{G}_{P}$ is given by
\begin{subequations}
\label{eqn_orthonal_projection}
\begin{align}
\Pi_{\mathcal{G}_{P}} &= \Pi_{G \mathcal{H}_{2}} := GG^{\star} \\
\iff \Pi_{\mathcal{G}^{\perp}_{P}} &= \Pi_{(G \mathcal{H}_{2})^{\perp}} := I - GG^{\star}.
\end{align}
\end{subequations}
Note that the projection operator $\Pi_{\mathcal{G}_{P}}$ is bounded. Let $G_{1}$ and $G_{2}$ denote the graph symbols of normalized RCFs of plants $P_{1}$ and $P_{2}$ respectively. The gap between $P_1$ and $P_2$ can be defined using three ways as follows:
\begin{enumerate}
  \item \emph{Projection-Based Definition:}
  This measures how far vectors in $\mathcal{G}_2$ stick out of $\mathcal{G}_1$, that is, how far they are not contained in $\mathcal{G}_1$:
  \begin{align}
  \label{eqn_gap_metric_definition_1}
  \delta_{g}(P_1, P_2) := \left\| \Pi_{\mathcal{G}_1^\perp} \Pi_{\mathcal{G}_2} \right\|.
  \end{align}
  
  \item \emph{Graph Operator Definition:}
  For an unit vector $\norm{u} = 1$, we see that $G_{2} u \in \mathcal{G}_{2}$ and $\Pi_{\mathcal{G}^{\perp}_{1}} G_{2} u \in \mathcal{G}^{\perp}_{1}$. Hence, the induced norm becomes:
  \begin{align}
  \label{eqn_gap_metric_definition_2}
  \delta_{g}(P_1, P_2) := \sup_{\|u\| = 1} \left\| \Pi_{\mathcal{G}_1^\perp} G_2 u \right\| = \left\| \Pi_{\mathcal{G}_1^\perp} G_2 \right\|.
  \end{align}
  
  \item \emph{Definition Using Norm of Projection Difference:}
  This measures the distance between the two projections in operator norm:
  \begin{align}
  \label{eqn_gap_metric_definition_3}
  \delta_{g}(P_1, P_2) = \left\| \Pi_{\mathcal{G}_1} - \Pi_{\mathcal{G}_2} \right\|.
  \end{align}
\end{enumerate}
The gap $\delta_{g}(P_1, P_2)$ can be calculated using Georgiou's formula (described in \cite{georgiou1988computation}) as follows:
\begin{subequations}
\label{eqn_gap_tryphon_formula}
\begin{align}
&\delta_{g}(P_1, P_2) 
= 
\max
\left\{
\vec{\delta}(P_1, P_2), \vec{\delta}(P_2, P_1)
\right\}, \quad \text{with} \label{eqn_gap_metric_formula} \\
&\vec{\delta}(P_1, P_2)
=
\norm{\Pi_{\mathcal{G}^{\perp}_{P_{2}}} \Pi_{\mathcal{G}_{P_{1}}}} =
\inf_{Q \in \mathcal{H}_{\infty}}
\norm{G_{1} - G_{2} Q}_{\mathcal{H}_{\infty}} \label{eqn_directed_gap_definition}.
\end{align}
\end{subequations}
Using perturbation theory for linear operators from \cite{kato2013perturbation}, for any closed subspaces $\mathcal{M}, \mathcal{N}$ of a Hilbert space, it holds that
\begin{align}
\left\| \Pi_{\mathcal{M}} - \Pi_{\mathcal{N}} \right\| = \max\left\{ \left\| \Pi_{\mathcal{M}^\perp} \Pi_{\mathcal{N}} \right\|, \left\| \Pi_{\mathcal{N}^\perp} \Pi_{\mathcal{M}} \right\| \right\}.
\end{align}
Since gap metric is symmetric meaning that $(\delta_{g}(P_1, P_2) = \delta_{g}(P_2, P_1))$, it suffices to consider just one direction. In the special case where $\mathcal{G}_1$ and $\mathcal{G}_2$ are graphs of stable systems with normalized coprime factorizations, by leveraging the gap being symmetric as described in \cite{green2012linear}, we get
\begin{align}
\left\| \Pi_{\mathcal{G}_1^\perp} \Pi_{\mathcal{G}_2} \right\| = \left\| \Pi_{\mathcal{G}_2^\perp} \Pi_{\mathcal{G}_1} \right\|.
\end{align}
The norm of this projection across all unit inputs tells us how misaligned the graphs are, which is what the gap metric captures. Hence, all three definitions of the gap metric $\delta_{g}(P_1, P_2)$ given by \eqref{eqn_gap_metric_definition_1}, \eqref{eqn_gap_metric_definition_2}, and \eqref{eqn_gap_metric_definition_3} are equal. We provide the following observations from \cite{vinnicombe2001uncertainty} as a proposition without proof. We note here that we will be using the results of the proposition later in our theoretical developments. 

\begin{proposition}(From \cite{vinnicombe2001uncertainty})   \label{proposition_vinnicombe} 
Given $\alpha \in (0, 1)$, let the controller $C_1$ stabilise the plant $P_1$ and further assume that the pair $(P_1, C_1)$ yields a performance measure of $b_{P_1, C_1} > \alpha$, where
\begin{subequations}
\label{eqn_b_pc_definition}
\begin{align}
b_{P_1,C_1}
&:=
\norm{Q(P_1, C_1)}^{-1}_{\mathcal{H}_{\infty}} \quad \text{and}, \\
Q(P_1, C_1)
&:=
\begin{bmatrix}
P_1 \\ I    
\end{bmatrix}
(I - C_1 P_1)^{-1}
\begin{bmatrix}
-C_1 & I    
\end{bmatrix}
\end{align}
\end{subequations}
Then, $C_1$ also stabilises the set of all plants given by
\begin{align}
\mathcal{B}_{\alpha}(P_{1})
:=
\left\{
P \in \mathcal{H}_{\infty} \mid \delta_{g}(P, P_1) \leq \alpha
\right\}
\end{align}
and $\forall P \in \mathcal{B}_{\alpha}(P_{1})$, the performance of $C_1$ degrades as
\begin{align}
\label{eqn_performance_degradation}
b_{P, C_1} \geq b_{P_1, C_1} - \delta_{g}(P, P_1).
\end{align}
\end{proposition}
\noindent Authors in \cite{cantoni2002linear} give the following closed loop deviation result  
\begin{align}
\label{eqn_performance_bounds_Q}
\norm{Q(P_1, C_1) - Q(P, C_1)}_{\mathcal{H}_{\infty}}
\leq
\frac{\delta_{g}(P, P_1)}{b_{P, C_1} \cdot b_{P_1, C_1}}.
\end{align}
Similarly, from \cite{tryhon_smith_tac_1990}, we observe that when $\delta_{g}(P, P_1) < 1$, 
\begin{align}
\label{eqn_performance_bounds_T}
&\norm{\mathbf{T}_{zw}(P, C_1)}_{\mathcal{H}_{\infty}}
\leq
\frac{\norm{\mathbf{T}_{zw}(P_1, C_1)}_{\mathcal{H}_{\infty}} + \delta_{g}(P, P_1)}{1 - \delta_{g}(P, P_1)}.
\end{align}
The essence of the above proposition is that when the plant model $P_1$ is known and when a corresponding stabilising controller $C_1$ is also known, it informs us apriori how $C_1$ will fair when applied to other plant models in the vicinity of the plant model $P_1$ given by plants $P \in \mathcal{B}_{\alpha}(P_{1})$. Particularly, the proposition informs us how the performance measure degrades followed by achievable $\mathcal{H}_{\infty}$ performance degradation in terms of the gap. In the following section, we will formulate the setting where the gap $\delta_{g}(P, P_1)$ becomes random and that will open a new perspective on doing PRC. Some motivating problems on why the gap $\delta_{g}(P, P_1)$ may become random and why to develop this theory are available in \cite{venkat_cdc_2024}. 
\section{Problem Formulation}
\label{sec_problem_formulation}
\subsection{Uncertain Dynamical System}
Consider the nominal model of the continuous time LTI dynamical system of the following form:
\begin{align}
\label{eqn_nominal_dynamics_system}
\bar{\Sigma}
:
\Bigl\{
\dot{x} 
=
A x + B u, \quad y = Cx,
\end{align}
where we refer to the system states as $x \in \mathbb{R}^{n}$ and the control inputs to the system as $u \in \mathbb{R}^{m}$, and system outputs as $y \in \mathbb{R}^{l}$ and the matrices $(A, B, C)$ are of appropriate dimensions. Real-world dynamical systems usually have some form of uncertainties associated with them either due to the lack of modelling tools or due to the inaccuracies of the modelling framework. Hence, in practise, all systems have inherent uncertainties affecting their evolution. We model the uncertainty affecting the evolution of such uncertain systems using $\theta \in \mathbb{R}^{p}$ with $p \leq (n+m+l)$ and $\theta$ directly affects the evolution of the perturbed system described as follows:
\begin{align}
\label{eqn_perturbed_dynamics_system}
\tilde{\Sigma}(\theta)
:
\Bigl\{
\dot{x} 
= 
A(\theta) x + B(\theta) u, \quad y = C(\theta) x.    
\end{align}
We will assume that $\theta \sim \mathbf{f_{\theta}}$ where, $\mathbf{f_{\theta}}$ denotes the distribution of the parameter $\theta$. For instance, we can assume that $\mathbf{f_{\theta}}$ is unknown but is believed to be belonging to a moment based ambiguity set $\mathcal{P}^{\theta}$ consistent with mean $\mu_{\theta} \in \mathbb{R}^{p}$ and covariance $\sigma^{2}_{\theta} I_{p} \succ 0$. However, for the ease of exposition, we will assume that $\theta \sim \mathbf{f}_{\theta} = \mathcal{N}(\mu_{\theta}, \sigma^{2}_{\theta} I_{p})$ as this will aid the formulation of the associated gap to become sub-Gaussian which is favourable for obtaining bounds on tail probability (as a Lipschitz function of a Gaussian random variable is sub-Gaussian \cite{vershynin2018high}). Further, we note here that $\tilde{\Sigma}(\theta_{0}) = \bar{\Sigma}$ meaning that 
the uncertainties of the perturbed system vanish at $\theta = \theta_{0}$ and the resulting system equals the nominal system (in sense of Figure \ref{fig:robustcontroldiagram} with $\Delta = 0$). This does not imply that $\mu_{\theta} = \theta_{0}$. The only nominal and valid requirement that is needed is that $\theta_{0} \in \mathbf{f}_{\theta}$ (perfectly fine even if the containment happens asymptotically) so that when uncertainties of the perturbed system vanish, it results in the nominal system. It would be interesting to investigate the randomness in uncertainties by taking into account their structural information as done in \cite{lanzon2009distance}, but we reserve that research direction for future work. While we don't take into account the structural information about the uncertainties in this research, we establish guarantees by incorporating further knowledge as to where in the model ambiguity set the system model is more likely to be in the space of (LTI) dynamical systems. We believe that by leveraging this additional knowledge on the space of dynamical systems, we can create new perspective for the field of PRC theory.

\subsection{Gap Between Nominal \& Perturbed   Models}
Having defined the evolution of the nominal model $\bar{\Sigma}$ using \eqref{eqn_nominal_dynamics_system} and perturbed model $\tilde{\Sigma}(\theta)$ in \eqref{eqn_perturbed_dynamics_system} for the system $\Sigma$, we denote the closed loop complementary sensitivity transfer function (from $w$ to $z$ in Figure \ref{fig:robustcontroldiagram}) of the system $\Sigma$ as $\mathbf{T}$. Then, the gap between the nominal model of the system $\bar{\Sigma}$ and the perturbed model of the system $\tilde{\Sigma}(\theta)$ denoted by $\mathrm{Gap}(\theta)$ can be defined using \eqref{eqn_gap_metric_definition_3} as
\begin{align}
\label{eqn_gap_between_models}    
\mathrm{Gap}(\theta) 
:=
\delta_{g}(\bar{\Sigma}, \tilde{\Sigma}(\theta)).
\end{align}
Note that for a fixed $\theta$ value, $\mathrm{Gap}(\theta)$ can be computed using \eqref{eqn_gap_tryphon_formula}. Note that for every realization of $\theta$ say $\bar{\theta}$ from $\mathbf{f}_{\theta}$, we get a deterministic graph subspace $\mathcal{G}_{\tilde{\Sigma}(\bar{\theta})}$ and hence a deterministic projection $\Pi_{\mathcal{G}_{\tilde{\Sigma}(\bar{\theta})}}$. But in general, the randomness in $\tilde{\Sigma}(\theta)$ due to $\theta$ manifests itself as the randomness in the graph subspace $\mathcal{G}_{\tilde{\Sigma}(\theta)}$ and this results in the corresponding projection operator $\Pi_{\mathcal{G}_{\tilde{\Sigma}(\theta)}}$ becoming random as well. Specifically, the projection $\Pi_{\mathcal{G}_{\tilde{\Sigma}(\theta)}}$ is a random operator-valued function of $\theta$. While we certainly want to investigate the above randomness of the projection operator and hence the randomness of the associated $\mathrm{Gap}(\theta)$ using \eqref{eqn_gap_metric_definition_1}, \eqref{eqn_gap_metric_definition_2}, and \eqref{eqn_gap_metric_definition_3} in detail, we reserve that exciting research direction as a future work. In this manuscript, we will aim to establish the fact that $\mathrm{Gap}(\theta)$ is a sub-Gaussian random variable and leverage the tools from high dimensional statistics to give probabilistic guarantees. Let $\Theta \subseteq \mathbb{R}^{p}$ denote the set of all possible values of $\theta$ such that, $\theta \sim \mathcal{N}(\mu_{\theta}, \sigma^{2}_{\theta} I_{p})$. We now state the main problems of interests that are being addressed in this manuscript.

\begin{problem} \label{problem_random_gap}
Since the $\mathrm{Gap}(\theta)$ is random, 
for a given nominal system model $\bar{\Sigma}$ and the perturbed system model $\tilde{\Sigma}(\theta)$ with $\theta \sim \mathcal{N}(\mu_{\theta}, \sigma^{2}_{\theta} I_{p})$, estimate the following quantities of interests:
\begin{enumerate}
    \item probability that the $\mathrm{Gap}(\theta)$ exceeds the given threshold $\epsilon \in (0,1)$ denoted by $\mathbb{P}(\mathrm{Gap}(\theta) \geq \epsilon)$.
    \item bound on the $\mathbb{E}[\mathrm{Gap}(\theta)]$.
\end{enumerate}
\end{problem}

\begin{problem} \label{problem_Hinf_performance}
Given nominal system model $\bar{\Sigma}$ and the perturbed system model $\tilde{\Sigma}(\theta)$ with $\theta \sim \mathcal{N}(\mu_{\theta}, \sigma^{2}_{\theta} I_{p})$ resulting in $\mathrm{Gap}(\theta)$ becoming random, address the following questions:
\begin{enumerate}
    \item how to guarantee for a given desired performance level $\gamma > 0$ and a violation probability $\beta \in (0,1)$ that 
    \begin{align}
    \mathbb{P}
    \left(
    \norm{\mathbf{T}(\tilde{\Sigma}(\theta))}_{\mathcal{H}_{\infty}} > \gamma
    \right)
    \leq 
    \beta.
    \end{align}
    \item obtain a bound for $\mathbb{E}\left[ \norm{\mathbf{T}(\tilde{\Sigma}(\theta))}_{\mathcal{H}_{\infty}} \right]$. 
\end{enumerate}

\end{problem}


\section{Solution Methodology} \label{sec_main}
Let us denote the normalized RCFs of the nominal model $\bar{\Sigma}$ and the perturbed model $\tilde{\Sigma}(\theta)$ of the system as $(\bar{N}, \bar{D})$ and  $(\tilde{N}(\theta), \bar{D}(\theta))$ respectively and further their respective graph symbols by $\bar{G}$ and $\tilde{G}(\theta)$. We begin the discussion in this section by analysing how should one understand the randomness associated with the normalised RCFs $\tilde{G}(\theta)$. Most of the work presented in the section would work for the case of $\mathrm{Gap}(\theta) < 1$. The case of $\mathrm{Gap}(\theta) = 1$ corresponds to unstable perturbed plant $\bar{\Sigma}(\theta)$ and we do not consider such cases ($\nu$-Gap is equipped to handle such cases and it is clearly out of the scope of this paper). Similarly, all the developments in this paper shall assume (unless otherwise specified) that the nominal plant $\bar{\Sigma}$ and the nominal controller $\bar{C}$ belong to the $\mathcal{H}_{\infty}$ space. Though depending upon the strength of the perturbation, the perturbed plant $\tilde{\Sigma}(\theta)$ may become open-loop stable or unstable and hence it may or may not belong to the $\mathcal{H}_{\infty}$ space, we only consider the ones that are in the $\mathcal{H}_{\infty}$ space so that $\mathrm{Gap}(\theta) < 1$. Though, we can write analogous theorem and lemma statements with guarantees in terms of the Left co-prime factors (LCFs), we will be sticking to the RCFs based statements without loss of generality in this manuscript. 

\subsection{About the Randomness of Graph Operator $\tilde{G}(\theta)$}
We begin our discussion about the randomness associated with the graph operator of the perturbed plant $\tilde{G}(\theta)$ where the uncertainty stems from the associated uncertainty of the parameter $\theta$. We have the following assumption in place concerned with the dependence of $\tilde{G}(\theta)$ on $\theta$. 

\begin{assumption} \label{assume_theta_continuity}
The mapping $\theta \mapsto \tilde{G}(\theta) := \begin{bmatrix} \tilde{N}(\theta) \\ \bar{D}(\theta) \end{bmatrix}$ is Fr\'echet differentiable meaning that $\tilde{N}(\theta), \tilde{D}(\theta)$ are continuously differentiable in $\theta$.
\end{assumption}

The following lemma analyses the case when $\mathbf{f}_{\theta}$ is a Gaussian distribution\footnote{In principle, any $\mathbf{f}_{\theta}$ with mean $\mu_{\theta}$ and covariance $\sigma^{2}_{\theta} I_{p}$ would suffice. However, it would need more investigation further down the line for giving probabilistic guarantees. So, Gaussian distribution is preferred for the ease of exposition.} and shows that the B\'ezout identity still holds under Gaussian parameter uncertainty. 

\begin{lemma}
\label{lemma_Bezout_Gaussian}
Given $\theta \sim \mathcal{N}(\mu_{\theta}, \sigma^{2}_{\theta} I_{p})$ with known mean $\mu_{\theta} \in \mathbb{R}^{p}$ and covariance $\sigma^{2}_{\theta} I_{p} \succ 0$, let $\tilde{\Sigma}(\theta) = N(\theta) D(\theta)^{-1} \in \mathbf{R}\mathcal{H}_{\infty}$ such that $\forall \theta \in \mathbb{R}^{p}$:
\begin{itemize}
    \item $ N(\theta), D(\theta) \in \mathcal{H}_{\infty} $ form a normalized RCF,
    \item $\theta \mapsto [N(\theta), D(\theta)] $ is continuous in the $ \mathcal{H}_{\infty} $-norm.
\end{itemize}
Then, $\exists X(\theta), Y(\theta) \in \mathcal{H}_{\infty}$ such that the B\'ezout identity
\begin{align}
\label{eqn_probab_Bezout_identity}    
X(\theta) N(\theta) + Y(\theta) D(\theta) = I
\end{align}
holds almost surely. That is, $\mathbb{P} \left( \theta \in \mathbb{R}^{p} \mid \eqref{eqn_probab_Bezout_identity} \text{ holds} \right) = 1$.
\end{lemma}

\begin{proof}
Let $\mathcal{B}_{\theta}$ denote the set of all $\theta \in \mathbb{R}^{M}$ such that  $\exists X(\theta), Y(\theta) \in \mathcal{H}_{\infty}$ and \eqref{eqn_probab_Bezout_identity} holds. Then, $\mathcal{B}^{c}_{\theta}$ will consists of those $\theta$ for which the coprime condition fails. This happens only when:
\begin{enumerate}
    \item $ D(\theta) $ is non-invertible in $ \mathcal{H}_{\infty} $, or
    \item There exists a common unstable factor between $ N(\theta) $ and $ D(\theta) $.
\end{enumerate}
These are pathological cases corresponding to a closed, nowhere-dense analytic subset of $\mathbb{R}^{M}$ with Lebesgue measure zero, since the set of parameters for which pole-zero cancellations occur on the imaginary axis is closed and of measure zero. Now, since $\theta \sim \mathcal{N}(\mu_{\theta}, \sigma^{2}_{\theta} I_{p})$ is absolutely continuous with respect to the Lebesgue measure, it assigns zero probability to any null Lebesgue-measure set. Hence, $\mathbb{P}(\theta \in \mathcal{B}^{c}_{\theta}) = 0 \iff  \mathbb{P}(\theta \in \mathcal{B}_{\theta}) = 1$. Hence, with probability one, the B\'ezout identity \eqref{eqn_probab_Bezout_identity} holds and
$X(\theta), Y(\theta) \in \mathcal{H}_{\infty}$ can be constructed using standard right coprime factorization algorithms such as extended Euclidean algorithm or inner-outer factorizations.
\end{proof}
\noindent \textbf{Remarks:} The randomness in the parameter $\theta$ induces randomness in the coprime factors $[N(\theta), D(\theta)]$ (which we studied in Lemma \ref{lemma_Bezout_Gaussian}). Subsequently, the randomness in the coprime factors $[N(\theta), D(\theta)]$ manifests itself as variations in the graph $\tilde{G}(\theta)$, and hence in the angle between the graph subspaces $\mathcal{G}_{\bar{\Sigma}}$ and $\mathcal{G}_{\tilde{\Sigma}(\theta)}$ which then finally leads to the randomness in the associated gap $\mathrm{Gap}(\theta)$. That is, randomness in $\theta \mapsto \tilde{G}(\theta) \mapsto \angle (\mathcal{G}_{\bar{\Sigma}}, \mathcal{G}_{\tilde{\Sigma}(\theta)}) \mapsto \mathrm{Gap}(\theta)$. \\

We have the following assumption to deal with the randomness in the graph $\tilde{G}(\theta)$ associated with the RCFs of the perturbed system $\tilde{\Sigma}(\theta)$.

\begin{assumption} \label{assume_graph_operator_Lipschitz}
The mapping $\theta \mapsto \tilde{G}(\theta)$ is Lipschitz in terms of the $\mathcal{H}_{\infty}$ norm with constant $\mathbf{L}_{\tilde{G}} > 0$, such that
\begin{align}
\label{eqn_Lipchitz_graph_perturbed_system}
\norm{\tilde{G}(\theta) - \tilde{G}(\theta^{\prime})}_{\mathcal{H}_{\infty}}
\leq 
\mathbf{L}_{\tilde{G}} \norm{\theta - \theta^{\prime}}, \quad \forall \theta, \theta^{\prime} \in \Theta.
\end{align}
\end{assumption}
From assumption \ref{assume_graph_operator_Lipschitz}, we immediately see that the projection operator $\Pi_{\mathcal{G}_{\tilde{\Sigma}(\theta)}}$ also continuously varies with respect to $\theta$. Hence, it's operator norm would be well-defined. Having mentioned all the assumptions needed, we will begin our study about giving probabilistic guarantees for robust performance and robust stability in terms of the gap metric.  

\subsection{Probabilistic Guarantee for Coprime Factor Uncertainty}
Coprime factor uncertainty can be understood as a combination of multiplicative and inverse multiplicative type uncertainties and the trade off between them is determined by the nominal plant. As a precursor to the gap metric, we will first demonstrate how the randomness in the co-prime factor uncertainty affects the robust stability associated with the nominal controller stabilising the nominal plant. The following theorem (probabilistic extension of Theorem 1 in \cite{tryhon_smith_tac_1990}) gives probabilistic guarantees on the nominal controller stabilising the perturbed plant under random coprime factor uncertainty.

\begin{theorem}
\label{theorem_coprime_probabilistic_guarantee}
Let $\bar{\Sigma} = ND^{-1} \in \mathcal{H}_{\infty}$ with a normalized RCF $(N,D)$, and let $\bar{C} \in \mathcal{H}_{\infty}$ be a controller that stabilizes $\bar{\Sigma}$ and results in  $b_{\bar{\Sigma}, \bar{C}} > 0$. Let assumption \ref{assume_graph_operator_Lipschitz} hold true for $\theta \sim \mathcal{N}(\mu_{\theta}, \sigma^{2}_{\theta} I_p)$. Let the randomly perturbed systems $\tilde{\Sigma}(\theta) = N(\theta) M(\theta)^{-1}$ be defined by normalized RCFs as
\begin{align*}
\tilde{G}(\theta)
:= 
\begin{bmatrix}
N(\theta) \\ D(\theta)    
\end{bmatrix}    
= 
\begin{bmatrix}
N + \Delta_N(\theta) \\ 
D + \Delta_D(\theta)    
\end{bmatrix}.    
\end{align*}
Let the random perturbation affecting  $\tilde{\Sigma}(\theta)$ be denoted by $\Delta(\theta) := \begin{bmatrix} \Delta_N(\theta) & \Delta_D(\theta) \end{bmatrix}$. If $\exists \mathbf{L}_{\Delta} > 0$ such that 
\begin{align}
\label{eqn_theorem_coprime_probabilistic_guarantee_condition}
\norm{\Delta(\theta)}_{\mathcal{H}_{\infty}} 
\leq 
\mathbf{L}_{\Delta}
\norm{\theta - \mu_{\theta}},
\end{align}
then $\mathbb{P}\left(\bar{C} \text{ stabilizes } \tilde{\Sigma}(\theta) \right)$ satisfies
\begin{align}
\label{eqn_theorem_coprime_probabilistic_guarantee}
\mathbb{P}\left(\norm{\Delta(\theta)}_{\mathcal{H}_{\infty}} < b_{\bar{\Sigma}, \bar{C}} \right) \geq 1 - \exp\left( - \frac{b_{\bar{\Sigma}, \bar{C}}^{2}}{2 \mathbf{L}_{\Delta}^{2} \sigma^{2}_{\theta}} \right).
\end{align}
\end{theorem}

\begin{proof}
From \cite{vidyasagar1984graph} and Theorem 1 of \cite{tryhon_smith_tac_1990}, we know that the nominal controller $\bar{C}$ will stabilize all perturbed plants $\tilde{\Sigma}(\theta) = N(\theta) D(\theta)^{-1}$ when the normalized RCFs satisfy
\begin{align}
\label{eqn_theorem_coprime_probabilistic_guarantee_interim_step_1}
\norm{\Delta(\theta)}_{\mathcal{H}_{\infty}} 
=
\norm{\begin{bmatrix} N(\theta) - N \\ D(\theta) - D \end{bmatrix} }_{\mathcal{H}_{\infty}} 
< 
b_{\bar{\Sigma}, \bar{C}}.
\end{align}
From theorem assumption \eqref{eqn_theorem_coprime_probabilistic_guarantee_condition}, we know that 
\begin{align*} 
\norm{\Delta(\theta)}_{\mathcal{H}_{\infty}} 
\leq 
\mathbf{L}_{\Delta}
\norm{\theta - \mu_{\theta}}.
\end{align*}
Note that 
\begin{align*}
\theta \sim \mathcal{N}(\mu_{\theta}, \sigma^{2}_{\theta} I_p)
&\iff 
\theta - \mu_{\theta} 
\sim 
\mathcal{N}(0, \sigma^{2}_{\theta} I_p) \\
&\iff
Z 
:=
\frac{\theta - \mu_{\theta}}{\sigma_{\theta}} 
\sim
\mathcal{N}(0, I_p)
\end{align*}
Then, $\norm{\theta - \mu_{\theta}} = \sigma_{\theta} \norm{Z}$. We know that the Euclidean norm of a standard Gaussian vector $Z \sim \mathcal{N}(0, I_p)$ denoted by $\norm{Z}$ is sub-Gaussian and it concentrates around $\sqrt{p}$. That is, $\forall \epsilon > 0$, we see that $\mathbb{P}(Z - \sqrt{p} \geq \epsilon) \leq e^{-\frac{\epsilon^{2}}{2}}$ meaning that $\norm{Z} \sim \chi_{p}$. Immediately, we see that $\norm{\theta - \mu_{\theta}} = \sigma_{\theta} \norm{Z}$ is sub-Gaussian with parameter $\sigma_{\theta}$. Similarly, it implies from \eqref{eqn_theorem_coprime_probabilistic_guarantee_condition} that $\norm{\Delta(\theta)}_{\mathcal{H}_{\infty}}$ is a Lipschitz function of a sub-Gaussian vector and hence it is sub-Gaussian with parameter $\mathbf{L}_{\Delta} \sigma_{\theta}$. Then, we know from \cite{vershynin2018high} that $\norm{\Delta(\theta)}_{\mathcal{H}_{\infty}}$ satisfies
\begin{align}
\label{eqn_theorem_coprime_probabilistic_guarantee_interim_step_2}
\mathbb{P}(\norm{\Delta(\theta)}_{\mathcal{H}_{\infty}} 
\geq 
b_{\bar{\Sigma}, \bar{C}}) 
\leq 
\exp\left( - \frac{b_{\bar{\Sigma}, \bar{C}}^{2}}{2 \mathbf{L}_{\Delta}^{2} \sigma^{2}_{\theta}} \right).
\end{align}
Writing the complement of \eqref{eqn_theorem_coprime_probabilistic_guarantee_interim_step_2} yields \eqref{eqn_theorem_coprime_probabilistic_guarantee} and combining with the condition for stability \eqref{eqn_theorem_coprime_probabilistic_guarantee_interim_step_1} completes the proof.
\end{proof}

Theorem \ref{theorem_coprime_probabilistic_guarantee} says that when the random perturbation affecting the nominal system turns out to be Lipschitz and sub-Gaussian, then it is possible to give probabilistic guarantees on the nominal controller stabilising the random perturbed system model. Note that $\Delta(\theta)$ essentially captures the gap between the nominal and the perturbed model of the system. Having studied the probabilistic guarantees corresponding to the randomness in coprime factor uncertainty description, we now turn our attention to study about the randomness in the more general setting involving the $\mathrm{Gap}(\theta)$ quantity. 

\subsection{Inferring the Lipschitz Constant of $\mathrm{Gap}(\theta)$}
When $\mathrm{Gap}(\theta)$ is random, it essentially reflects the uncertainty in how the two graph subspaces $\mathcal{G}_{\bar{\Sigma}}, \mathcal{G}_{\tilde{\Sigma}(\theta)}$ are aligned with each other. Towards this perspective, we will now study about the Lipschitz constant associated with the gap $\mathrm{Gap}(\theta)$ in the following theorem.

\begin{theorem} \label{theorem_gap_lipschitz}
Let assumptions \ref{assume_theta_continuity} and \ref{assume_graph_operator_Lipschitz} hold true for the nominal model $\bar{\Sigma}$ and perturbed model $\tilde{\Sigma}(\theta)$ of the dynamical system $\Sigma$. Then, $\exists \mathbf{L}_{\mathrm{gap}} > 0$ that depends on $\bar{\Sigma}$ and $\mathbf{L}_{\tilde{G}}$ such that $\forall \theta, \theta^{\prime} \in \Theta$, the gap $\mathrm{Gap}(\theta)$ is $\mathbf{L}_{\mathrm{gap}}$-Lipschitz. That is,
\begin{align}
\label{eqn_gap_Lipschitz}
\left| \mathrm{Gap}(\theta) - \mathrm{Gap}(\theta^{\prime}) \right| \leq \mathbf{L}_{\mathrm{gap}} \norm{\theta - \theta^{\prime}}.
\end{align}
\end{theorem}

\begin{proof}
Using \eqref{eqn_gap_metric_definition_2}, we see 
\begin{align}
\label{eqn_theorem_gap_lipschitz_interim_1}
\mathrm{Gap}(\theta) 
=
\norm{\Pi_{\mathcal{G}^{\perp}_{\bar{\Sigma}}} \tilde{G}(\theta)}
\end{align}
By assumptions \ref{assume_theta_continuity} and \ref{assume_graph_operator_Lipschitz}, $\tilde{G}(\theta)$ varies continuously with $\theta$, with Lipschitz constant $\mathbf{L}_{\tilde{G}} > 0$ as described by \eqref{eqn_Lipchitz_graph_perturbed_system}. Then, using \eqref{eqn_theorem_gap_lipschitz_interim_1}, $\forall \theta, \theta^{\prime} \in \Theta$, we see that
\begin{align*}
\left| \mathrm{Gap}(\theta) - \mathrm{Gap}(\theta^{\prime}) \right|    
&= 
\left|
\norm{\Pi_{\mathcal{G}^{\perp}_{\bar{\Sigma}}} \tilde{G}(\theta)} - \norm{\Pi_{\mathcal{G}^{\perp}_{\bar{\Sigma}}} \tilde{G}(\theta^{\prime})}
\right| \\
&\leq
\norm{\Pi_{\mathcal{G}^{\perp}_{\bar{\Sigma}}} \left(\tilde{G}(\theta) - \tilde{G}(\theta^{\prime})\right)} \\
&\leq 
\norm{\Pi_{\mathcal{G}^{\perp}_{\bar{\Sigma}}}} \norm{\tilde{G}(\theta) - \tilde{G}(\theta^{\prime})}_{\mathcal{H}_{\infty}} \\
&\leq 
\underbrace{\norm{\Pi_{\mathcal{G}^{\perp}_{\bar{\Sigma}}}} 
\mathbf{L}_{\tilde{G}}}_{:= \mathbf{L}_{\mathrm{gap}}} \norm{\theta - \theta^{\prime}}.
\end{align*}
Here, we have used the reverse triangle inequality to get the first inequality above. We have used the fact that projections are bounded linear operators to get the second inequality along with the fact that $\mathcal{H}_{\infty}$ norm also happens to be equal to the induced $\mathcal{L}_{2}$ norm. Finally, we used \eqref{eqn_Lipchitz_graph_perturbed_system} to get the final inequality. Note that the Lipschitz constant $\mathbf{L}_{\mathrm{gap}}$ explicitly depends upon the norm of the projection operator onto the graph subspace of $\bar{\Sigma}$ and the Lipschitz constant $\mathbf{L}_{\tilde{G}}$ from \eqref{eqn_Lipchitz_graph_perturbed_system} and this completes the proof of the theorem.
\end{proof}

In the following corollary and in the next section, we now address the Problem \ref{problem_random_gap} stated before. We know from \cite{vershynin2018high} that any Lipschitz function of a Gaussian random vector is sub-Gaussian. That is, with $\theta \sim \mathcal{N}(\mu_{\theta}, \sigma^{2}_{\theta} I_{p})$ and \eqref{eqn_gap_Lipschitz}, we immediately see that $\mathrm{Gap}(\theta)$ is sub-Gaussian with parameter $\sigma_{\theta} \mathbf{L}_{\mathrm{gap}}$ and corollary \ref{corollary_prob_gap_exceed_threshold} that is stated below formally establishes that observation.

\begin{corollary}\label{corollary_prob_gap_exceed_threshold}
Let the assumptions of the Theorem \ref{theorem_gap_lipschitz} be true. If $\theta \sim \mathcal{N}(\mu_{\theta}, \sigma^{2}_{\theta} I_{p})$, then $\mathrm{Gap}(\theta)$ is sub-Gaussian with parameter $\sigma_{\theta} \mathbf{L}_{\mathrm{gap}}$. That is, $\forall \epsilon > 0$, we see that 
\begin{align}
\label{eqn_corollary_prob_gap_exceed_threshold}
\mathbb{P}
\left(
\mathrm{Gap}(\theta) - \mathbb{E}[\mathrm{Gap}(\theta)] \geq \epsilon
\right)
\leq
e^{-\frac{\epsilon^{2}}{2\sigma^{2}_{\theta} \mathbf{L}_{\mathrm{gap}}^{2}}}.
 \end{align}
\end{corollary}
\begin{proof}
We know from Gaussian concentration inequality in \cite{vershynin2018high} that if $f: \mathbb{R}^{n} \rightarrow \mathbb{R}$ is Lipschitz with constant $C > 0$ and $\theta \sim \mathcal{N}(\mu_{\theta}, \sigma^{2}_{\theta} I_{p})$, then $\forall \epsilon > 0$
\begin{align}
\label{eqn_Gaussian_conc_inequality}
\mathbb{P}
\left(
f(\theta) - \mathbb{E}[f(\theta)] \geq \epsilon
\right)
\leq
e^{-\frac{\epsilon^{2}}{2\sigma^{2}_{\theta} C^{2}}}.
\end{align}
Substituting $f(\theta) = \mathrm{Gap}(\theta)$ with Lipschitz constant $\mathbf{L}_{\mathrm{gap}} > 0$ from \eqref{eqn_gap_Lipschitz}, we get \eqref{eqn_corollary_prob_gap_exceed_threshold}.
\end{proof}



\subsection{Inferring the Expected Value of $\mathrm{Gap}(\theta)$}
Given that $\theta \sim \mathcal{N}(\mu_{\theta}, \sigma^{2}_{\theta} I_{p})$, we leverage the fact that mean $\mu_{\theta}$ and covariance $\sigma^{2}_{\theta} I_{p}$ are deterministic known quantities to get an upper bound for the expected gap in terms of the them. We state the following proposition using the Jensen's inequality from \cite{vershynin2018high} which we will use later.

\begin{proposition} \label{proposition_Jensen_inequality}
Let $\theta \sim \mathcal{N}(\mu_{\theta}, \sigma^{2}_{\theta} I_{p})$. Then, 
\begin{subequations}
\label{eqn_jensen_results}
\begin{align}
\mathbb{E}[\norm{\theta - \mu_{\theta}}] 
&\leq 
\sqrt{ \sigma^{2}_{\theta} p}. \label{eqn_jensen_result_1}\\
\mathbb{E}[\norm{\theta}] 
&\leq 
\sqrt{ \sigma^{2}_{\theta} p + \norm{\mu_{\theta}}^{2}}. \label{eqn_jensen_result_2}
\end{align}    
\end{subequations}
\end{proposition}
\begin{proof} 
For a fixed $\theta^{\prime} \in \mathbb{R}^{p}$, Jensen inequality \cite{vershynin2018high} states that
\begin{align}
\label{eqn_Jensen_inequality}
\mathbb{E}[\|\theta - \theta^{\prime}\|] \leq \sqrt{ \mathbf{Tr}(\mathbf{Cov}(\theta)) + \|\mu_{\theta} - \theta^{\prime} \|^2 }.    
\end{align}    
Note that, $\mathbf{Tr}(\sigma^{2}_{\theta} I_{p}) = \sigma^{2}_{\theta} p$. Substituting $\theta^{\prime} = \mu_{\theta}$ into \eqref{eqn_Jensen_inequality}, we get \eqref{eqn_jensen_result_1}. On the other hand, substituting $\theta^{\prime} = 0$ into \eqref{eqn_Jensen_inequality}, we get \eqref{eqn_jensen_result_2}.
\end{proof}

In the following lemma, we get an upper bound on the expected gap using the Lipschitz continuity of the gap and Proposition \ref{proposition_Jensen_inequality}.

\begin{lemma} \label{lemma_expected_gap}
Let $\bar{\Sigma}$ be a known nominal system model and let $\tilde{\Sigma}(\theta) $ be a random LTI system model depending on $\theta \sim \mathcal{N}(\mu_{\theta}, \sigma^{2}_{\theta} I_{p})$. Further, suppose that the gap $\mathrm{Gap}(\theta)$ is $\mathbf{L}_{\mathrm{gap}}$-Lipschitz as in \eqref{eqn_gap_Lipschitz}. Then, $\forall \theta, \theta^{\prime} \in \mathbb{R}^{p}$,
\begin{align}
\label{eqn_lemma_expected_gap}
\mathbb{E}\left[ \mathrm{Gap}(\theta) \right] 
\leq 
\mathbb{E}\left[ \mathrm{Gap}(\theta^{\prime}) \right] 
+
\mathbf{L}_{\mathrm{gap}} 
\mathbb{E}\left[ \norm{\theta - \theta^{\prime}} \right].
\end{align}
\end{lemma}

\begin{proof}
We know from \eqref{eqn_gap_Lipschitz} that $\forall \theta, \theta^{\prime} \in \mathbb{R}^{p}$, 
\begin{align}
\mathrm{Gap}(\theta)
\leq 
\mathrm{Gap}(\theta^{\prime})
+ 
\mathbf{L}_{\mathrm{gap}} \norm{\theta - \theta^{\prime}}.
\end{align}
Taking expectations of both sides \& using linearity of expectation yields
\eqref{eqn_lemma_expected_gap}.
\end{proof}
Note that we need to study about $\mathbb{E}[\norm{\theta - \theta^{\prime}}]$ to give an estimate about the expected value of the $\mathrm{Gap}(\theta)$. In the following corollaries, we will use the Proposition \ref{proposition_Jensen_inequality} to find upper bounds for $\mathbb{E}[\norm{\theta - \theta^{\prime}}]$ in \eqref{eqn_lemma_expected_gap}.

\begin{corollary} \label{corollary_expected_gap_subGaussian}
Under the assumptions of Lemma \ref{lemma_expected_gap}, when $\theta^{\prime} = \mu_{\theta}$, we see
\begin{align}
\label{eqn_corollary_lemma_expected_gap}
\mathbb{E}\left[ \mathrm{Gap}(\theta) \right] 
\leq 
\underbrace{
\mathrm{Gap}(\mu_{\theta})
+
\mathbf{L}_{\mathrm{gap}} 
\sqrt{ \sigma^{2}_{\theta} p}
}_{:= C_{\mathrm{gap}}}. 
\end{align}
\end{corollary}

\begin{proof}
Note that $\theta^{\prime} = \mu_{\theta}$ is a valid assumption to make as $\mathbb{E}[\theta] = \mu_{\theta}$. On the other hand, we see that $\mathbb{E}\left[ \mathrm{Gap}(\mu_{\theta}) \right] = \mathrm{Gap}(\mu_{\theta})$ becomes a deterministic quantity. Then, using \eqref{eqn_jensen_result_1} in \eqref{eqn_lemma_expected_gap}, we get \eqref{eqn_corollary_lemma_expected_gap}.
\end{proof}

\begin{corollary}
Under the assumptions of the Lemma \ref{lemma_expected_gap} and when $\tilde{\Sigma}(\theta_{0}) = \bar{\Sigma}$, we see that  
\begin{align}
\label{eqn_lemma_expected_gap_simple}
\mathbb{E}\left[ \mathrm{Gap}(\theta) \right] 
\leq 
\mathbf{L}_{\mathrm{gap}} 
\sqrt{ \sigma^{2}_{\theta} p + \norm{\mu_{\theta} - \theta_{0}}^{2}}.
\end{align}
\end{corollary}
\begin{proof}
Note that $\tilde{\Sigma}(\theta_{0}) = \bar{\Sigma} \iff \mathrm{Gap}(\theta_{0}) = 0$. Substituting $\theta^{\prime} = \theta_{0}$ in \eqref{eqn_lemma_expected_gap} yields 
\begin{align}
\label{eqn_lemma_expected_gap_simple_interim_1}
\mathbb{E}\left[ \mathrm{Gap}(\theta) \right] 
\leq 
\mathbf{L}_{\mathrm{gap}} 
\mathbb{E}\left[ \norm{\theta - \theta_{0}} \right].
\end{align}
Using \eqref{eqn_jensen_result_2} in \eqref{eqn_lemma_expected_gap_simple_interim_1}, we get
\eqref{eqn_lemma_expected_gap_simple}.
\end{proof}

\subsection{Probabilistic Robust Stability Guarantees}
Having explored the bounds on the expected value of the $\mathrm{Gap}(\theta)$, we are interested in studying about a nominal controller's robust stability property while stabilising a nominal plant. Particularly, we will investigate
the probability of a nominal controller stabilising a random plant using the above obtained gap metric bounds. This will provide probabilistic controller certification under non-zero mean gap uncertainty.
\begin{theorem}
\label{theorem_probabilistic_stability}
Let $\bar{\Sigma} \in \mathcal{H}_{\infty}$ be the nominal plant, and let the nominal controller $\bar{C}$ stabilize $\bar{\Sigma}$ and results in $b_{\bar{\Sigma}, \bar{C}} > 0$. Let $\mathrm{Gap}(\theta)$ be the gap associated with $\bar{\Sigma}$ and a random perturbed plant $\tilde{\Sigma} \in \mathcal{H}_{\infty}$ such that $\mathrm{Gap}(\theta)$ is sub-Gaussian with mean $\mathbb{E}[\mathrm{Gap}(\theta)]$ and variance $\sigma^{2}_{\theta} \mathbf{L}^{2}_{\mathrm{gap}}$. If
\begin{align}
\label{eqn_theorem_probabilistic_stability_condition}
\mathbb{E}[\mathrm{Gap}(\theta)] < b_{\bar{\Sigma}, \bar{C}},
\end{align}
then, for the tolerance $\varepsilon_{\mathrm{tol}} := b_{\bar{\Sigma}, \bar{C}} - \mathbb{E}[\mathrm{Gap}(\theta)] > 0$, the probability that $\bar{C}$ stabilizes $\tilde{\Sigma}(\theta)$ satisfies
\begin{align}
\label{eqn_theorem_probabilistic_stability}
\mathbb{P}\left( \mathrm{Gap}(\theta) < b_{\bar{\Sigma}, \bar{C}} \right) \geq 1 - \exp\left( -\frac{\varepsilon_{\mathrm{tol}}^2}{2 \sigma^{2}_{\theta} \mathbf{L}^{2}_{\mathrm{gap}}} \right).
\end{align}
\end{theorem}

\begin{proof}
Since $\mathrm{Gap}(\theta)$ is sub-Gaussian with mean $\mu := \mathbb{E}[\mathrm{Gap}(\theta)]$ and variance $\sigma^{2}_{\theta} \mathbf{L}^{2}_{\mathrm{gap}}$, applying the standard sub-Gaussian tail bounds, we get for any given tolerance $\varepsilon_{\mathrm{tol}} > 0$,
\begin{align}
\label{eqn_theorem_probabilistic_stability_interim_step_0}
\mathbb{P}(\mathrm{Gap}(\theta) \geq 
\mathbb{E}[\mathrm{Gap}(\theta)] + \varepsilon_{\mathrm{tol}}) \leq \exp\left( -\frac{\varepsilon_{\mathrm{tol}}^{2}}{2 \sigma^{2}_{\theta} \mathbf{L}^{2}_{\mathrm{gap}}} \right).
\end{align}
Setting $\varepsilon_{\mathrm{tol}} := b_{\bar{\Sigma}, \bar{C}} - \mathbb{E}[\mathrm{Gap}(\theta)] > 0$ in \eqref{eqn_theorem_probabilistic_stability_interim_step_0}, we see that
\begin{align}
\label{eqn_theorem_probabilistic_stability_interim_step_1}
\mathbb{P}(\mathrm{Gap}(\theta) \geq b_{\bar{\Sigma}, \bar{C}}) \leq \exp\left( -\frac{(b_{\bar{\Sigma}, \bar{C}} - \mathbb{E}[\mathrm{Gap}(\theta)])^2}{2 \sigma^{2}_{\theta} \mathbf{L}^{2}_{\mathrm{gap}}} \right).
\end{align}
Since the controller $\bar{C}$ stabilizes $\bar{\Sigma}$ and results in $b_{\bar{\Sigma}, \bar{C}} > 0$, we know from Proposition \ref{proposition_vinnicombe} that $\bar{C}$ will also stabilise any plant $\tilde{\Sigma}$ if $\delta_{\mathrm{gap}}(\bar{\Sigma}, \tilde{\Sigma}) < b_{\bar{\Sigma}, \bar{C}}$. Hence, $\bar{C}$ will stabilize $\tilde{\Sigma}(\theta)$ if $\mathrm{Gap}(\theta) < b_{\bar{\Sigma}, \bar{C}}$. 
Then, the probability that $\mathrm{Gap}(\theta) < b_{\bar{\Sigma}, \bar{C}}$, that is, $\bar{C}$ stabilizes $\tilde{\Sigma}(\theta)$, is given by the complementary of \eqref{eqn_theorem_probabilistic_stability_interim_step_1}, which is equal to \eqref{eqn_theorem_probabilistic_stability} and the proof is complete.
\end{proof}

\subsubsection{Connecting Robust Stability \& Violation Probability}
If we want the stabilisation of $\tilde{\Sigma}(\theta)$ by nominal controller $\bar{C}$ to happen with a probability of at least $1 - \beta$, where $\beta \in (0,1)$ denotes the violation probability, then using Theorem \ref{theorem_probabilistic_stability}, we can find a corresponding condition on the expected value of the $\mathrm{Gap}(\theta)$ in terms of $\beta$. The following corollary formally establishes that result.

\begin{corollary}
\label{corollary_probabilistic_stability}    
Given a violation probability $\beta \in (0,1)$, suppose that 
\begin{align}
\label{eqn_corollary_probabilistic_stability}   
\mathbb{E}[\mathrm{Gap}(\theta)]
\leq
b_{\bar{\Sigma}, \bar{C}}
-
\mathbf{L}_{\mathrm{gap}} \sqrt{2 \sigma^{2}_{\theta} \log\left( \frac{1}{\beta} \right)}.
\end{align}
Then,
\begin{align}
\mathbb{P}
(\bar{C} \text{ stabilises } \tilde{\Sigma}(\theta)) 
\geq 
1 - \beta.
\end{align}
\end{corollary}
\begin{proof}
From \eqref{eqn_theorem_probabilistic_stability}, we know for $\varepsilon_{\mathrm{tol}} := b_{\bar{\Sigma}, \bar{C}} - \mathbb{E}[\mathrm{Gap}(\theta)] > 0$, the probability that $\bar{C}$ stabilizes $\tilde{\Sigma}(\theta)$ satisfies
\begin{align*}
\mathbb{P}\left( \mathrm{Gap}(\theta) < b_{\bar{\Sigma}, \bar{C}} \right) 
\geq 
1 - 
\exp\left( -\frac{\varepsilon_{\mathrm{tol}}^2}{2 \sigma^{2}_{\theta} \mathbf{L}^{2}_{\mathrm{gap}}} \right).
\end{align*}    
Set violation probability $\beta = \exp\left( -\frac{\varepsilon_{\mathrm{tol}}^2}{2 \sigma^{2}_{\theta} \mathbf{L}^{2}_{\mathrm{gap}}} \right)$ and solving for $\mathbb{E}[\mathrm{Gap}(\theta)]$, we get \eqref{eqn_corollary_probabilistic_stability} and this completes the proof.
\end{proof}

\subsubsection{Numerical Demonstration}
\begin{figure}
    \centering    \includegraphics[width=\linewidth]{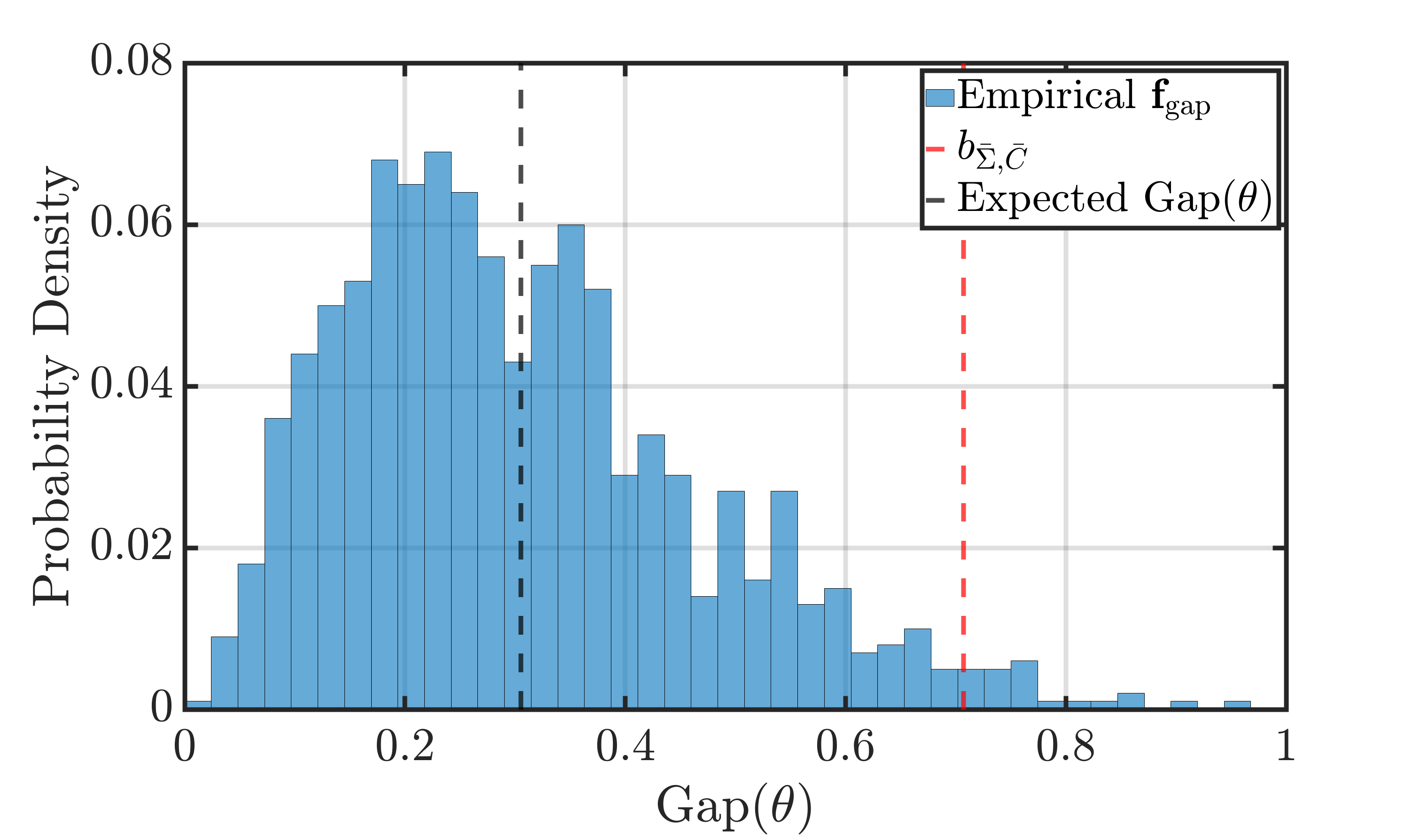}
    \caption{Results of Monte-carlo simulation with $N = 10^{4}$ independent trials validating claims of Theorem \ref{theorem_probabilistic_stability} are shown here. Clearly \eqref{eqn_theorem_probabilistic_stability_condition} is satisfied, and we see empirically that the probability that $\bar{C}$ stabilizes $\tilde{\Sigma}(\theta)$ is $0.9777$ which is clearly more than the lower bound of $0.5561$ obtained using \eqref{eqn_theorem_probabilistic_stability}.}
    \label{fig:probRobustStability}
\end{figure}
To demonstrate the probabilistic robust stability guarantees obtained in this subsection, we considered a nominal system $\bar{\Sigma}: (A, B, C, D) = (-1, 1, 1, 0)$. We sampled $N = 10^{4}$ values of the uncertain parameter $\theta \sim \mathcal{N}(\begin{bmatrix}
0.50 & -0.25
\end{bmatrix}, 0.25^{2} I_{2})$. The perturbed plants were formed as $\tilde{\Sigma}(\theta^{(i)}) = (-1 + \theta^{(i)}_{1}, 1, 1 + \theta^{(i)}_{2}, 0)$ for $i = 1, \dots, N$. After doing a Monte-carlo simulation using $N$ independent trials involving samples of $\theta$, the results are shown in Figure \ref{fig:probRobustStability}. Using \eqref{eqn_lemma_expected_gap_simple_interim_1}, we estimated $\mathbb{E}[\mathrm{Gap}(\theta)] = 0.3032$. The performance measure was $b_{\bar{\Sigma}, \bar{C}} = 0.7071$ for the nominal controller $\bar{C}$ that placed the closed loop poles at $-2$. Empirically, we found $\mathbb{P}(\mathrm{Gap}(\theta) < b_{\bar{\Sigma}, \bar{C}}) = 0.9777$ which was greater than the lower bound of $0.5561$ from \eqref{eqn_theorem_probabilistic_stability} and thereby validating the claims of Theorem \ref{theorem_probabilistic_stability}. 

\subsection{Probabilistic Closed Loop Deviation Guarantees}
Given that we have investigated the probabilistic guarantees of the gap $\mathrm{Gap}(\theta)$ greater than some threshold in Theorem \ref{theorem_gap_lipschitz} and probabilistic robust stability guarantees in Theorem \ref{theorem_probabilistic_stability}, we now turn our attention to give probabilistic closed loop deviation guarantees. Before proceeding further, we require an assumption on the Lipschitz continuity of $\norm{ Q(\tilde{\Sigma}(\theta), \bar{C}) }_{\mathcal{H}_{\infty}}$ with respect to parameter $\theta \sim \mathcal{N}(\mu_{\theta}, \sigma^{2}_{\theta} I_p)$.

\begin{assumption} \label{assume_Q_Lipschitz}
Given $\theta \sim \mathcal{N}(\mu_{\theta}, \sigma^{2}_{\theta} I_p)$, $\exists \mathbf{L_{Q}} > 0$ such that the function $\mathbf{f_{Q}}(\theta) := \norm{ Q(\tilde{\Sigma}(\theta), \bar{C}) }_{\mathcal{H}_{\infty}}$ is $\mathbf{L_{Q}}$-Lipschitz.
\end{assumption}

\textbf{Remarks:} This is a mild and valid assumption to make as usually the system matrices $A(\theta), B(\theta), C(\theta)$ of the perturbed model $\tilde{\Sigma}(\theta)$ depend smoothly on $\theta$ (are Fr\'echet differentiable in $\mathcal{H}_{\infty}$ norm with respect to $\theta$), and the function $\mathbf{f_{Q}}(\theta)$ is formed through algebraic and analytic operations on these matrices. Further, we had also earlier restricted our study to the case where the closed-loop remains internally stable over all realizations of $\theta \sim \mathcal{N}(\mu_\theta, \sigma_{\theta}^2 I_p)$. Hence, the map $\theta \mapsto \mathbf{f_{Q}}(\theta)$ is differentiable and hence Lipschitz on the space of $\theta$. Since Gaussian distributions concentrate their mass near the mean, any potential growth in the Lipschitz constant outside compact sets has negligible impact. Thus, global Lipschitz continuity of $\mathbf{f_{Q}}(\theta)$ is a conservative yet reasonable assumption to make and it enables rigorous probabilistic analysis. Since $\mathbf{f_{Q}}(\theta)$ is Lipschitz on $\theta \in \mathbb{R}^{p}$, and $\theta \sim \mathcal{N}(\mu_{\theta}, \sigma^{2}_{\theta} I_p)$, we infer from \cite{vershynin2018high} that it is also sub-Gaussian with parameter $\sigma_{\theta} \mathbf{L_{Q}} > 0$.

Using the following theorem, we will now give probabilistic closed loop deviation bound under the action of controller $\bar{C}$.
\begin{theorem} 
\label{theorem_Q_difference}
Let the given nominal model and nominal stabilising controller pair $(\bar{\Sigma}, \bar{C})$ achieve $b_{\bar{\Sigma}, \bar{C}} > 0$. Let $\tilde{\Sigma}(\theta)$ be a perturbed plant with $\theta \sim \mathcal{N}(\mu_{\theta}, \sigma^{2}_{\theta} I_p)$ such that $\mathrm{Gap}(\theta)$ satisfies \eqref{eqn_gap_Lipschitz} and also let assumption \ref{assume_Q_Lipschitz} to hold true. Given $\beta \in (0,1)$, define the following thresholds 
\begin{subequations}
\label{eqn_epsilons}    
\begin{align}
\varepsilon_{\mathrm{gap}} 
&:= 
\sqrt{2 \sigma^{2}_{\theta} \mathbf{L}_{\mathrm{gap}}^2 \log\left( \frac{2}{\beta} \right)}, \quad \text{and} \label{eqn_epsilon_gap_definition_theorem_Q}\\
\varepsilon_{\mathbf{Q}} 
&:= 
\sqrt{2 \sigma^{2}_{\theta} \mathbf{L_{Q}}^2 \log\left( \frac{2}{\beta} \right)} \label{eqn_epsilon_Q_definition_theorem_Q}.
\end{align}
\end{subequations}
Then, with probability at least $1 - \beta$, the following inequality on the closed loop deviation will hold:
\begin{equation}
\label{eqn_theorem_Q_difference}
\footnotesize
\norm{Q(\tilde{\Sigma}(\theta), \bar{C}) - Q(\bar{\Sigma}, \bar{C})}_{\mathcal{H}_{\infty}}
\leq \frac{ \left( \mathbb{E}[\mathrm{Gap}(\theta)] + \varepsilon_{\mathrm{gap}} \right)  \left( \mathbb{E}[\mathbf{f_{Q}}(\theta)] + \varepsilon_{\mathbf{Q}} \right) }{ b_{\bar{\Sigma}, \bar{C}} }.
\end{equation}
\end{theorem}

\begin{proof}
We are essentially looking for the probabilistic satisfaction of \eqref{eqn_performance_bounds_Q} inspired from Theorem III.2 of \cite{cantoni2002linear}. Given that $\mathrm{Gap}(\theta)$ is $\mathbf{L}_{\mathrm{gap}}$-Lipschitz from \eqref{eqn_gap_Lipschitz} and $\theta \sim \mathcal{N}(\mu_{\theta}, \sigma^{2}_{\theta} I_p)$, we see that under \eqref{eqn_epsilon_gap_definition_theorem_Q}, the concentration bound from \cite{vershynin2018high} implies that 
\begin{align*}
&\mathbb{P}\left( \mathrm{Gap}(\theta) \geq \mathbb{E}[\mathrm{Gap}(\theta)] + \varepsilon_{\mathrm{gap}} \right) \leq \frac{\beta}{2}.
\end{align*}
Note that above guarantee can also be arrived using \eqref{eqn_theorem_probabilistic_stability_interim_step_0} with $\varepsilon_{\mathrm{tol}} = \varepsilon_{\mathrm{gap}}$. That is, 
\begin{align*}
\mathbb{P}(\mathrm{Gap}(\theta) \geq 
\mathbb{E}[\mathrm{Gap}(\theta)] + \varepsilon_{\mathrm{gap}}) \leq \exp\left( -\frac{\varepsilon_{\mathrm{gap}}^{2}}{2 \sigma^{2}_{\theta} \mathbf{L}^{2}_{\mathrm{gap}}} \right) = \frac{\beta}{2}.
\end{align*}
Solving for $\varepsilon_{\mathrm{gap}}$ in terms of $\beta$, we get \eqref{eqn_epsilon_gap_definition_theorem_Q}. By similar arguments, we observe from assumption \ref{assume_Q_Lipschitz} that $\mathbf{f_{Q}}(\theta)$ is $\mathbf{L_{Q}}$-Lipschitz and hence is sub-Gaussian with parameter $\sigma_{\theta} \mathbf{L_{Q}}$. Then $\mathbf{f_{Q}}(\theta)$ satisfies the sub-Gaussian concentration bound
\begin{align*}
\mathbb{P}\left( \mathbf{f_{Q}}(\theta) > \mathbb{E}[\mathbf{f_{Q}}(\theta)] + \varepsilon_{\mathbf{Q}} \right) \leq \frac{\beta}{2},
\end{align*}
when the threshold $\varepsilon_{\mathbf{Q}}$ is given by 
\begin{align*}
\varepsilon_{\mathbf{Q}} := \sqrt{2 \sigma^{2}_{\theta} \mathbf{L_{Q}}^2 \log\left( \frac{2}{\beta} \right)}.    
\end{align*}
Now define the events 
\begin{align*}
\mathcal{G} 
&:= 
\left\{ \mathrm{Gap}(\theta) > \mathbb{E}[\mathrm{Gap}(\theta)] + \varepsilon_{\mathrm{gap}} \right\} \\
\mathcal{F} 
&:= 
\left\{ \mathbf{f_{Q}}(\theta) > \mathbb{E}[\mathbf{f_{Q}}(\theta)] + \varepsilon_{\mathbf{Q}} \right\}
\end{align*}
Then, our requirement $\mathbb{P}(\mathcal{G}^{c} \cap \mathcal{F}^{c}) \geq 1 - \beta$ because 
\begin{align*}
\mathbb{P}(\mathcal{G}^{c} \cap \mathcal{F}^{c}) 
&= 1 - \mathbb{P}(\mathcal{G} \cup \mathcal{F}) \\
&\geq
1 - \left(\mathbb{P}(\mathcal{G}) + \mathbb{P}(\mathcal{F}) \right) \\
&=
1 - \frac{\beta}{2} - \frac{\beta}{2} \\
&=
1 - \beta.
\end{align*}
The above inequality arises due to the application of the Boole's inequality in the second step.
From \eqref{eqn_performance_bounds_Q}, we get
\begin{align*}
\norm{Q(\tilde{\Sigma}(\theta), \bar{C}) - Q(\bar{\Sigma}, \bar{C})}_{\mathcal{H}_{\infty}}
\leq 
\frac{ \mathrm{Gap}(\theta) \, Q(\tilde{\Sigma}(\theta), \bar{C}) }
{ b_{\bar{\Sigma}, \bar{C}} }.
\end{align*}
Since $\mathbb{P}(\mathcal{G}^{c} \cap \mathcal{F}^{c}) \geq 1 - \beta$, we simultaneously substitute both the numerator terms by the respective upper bounds from $\mathcal{G}^{c}$ and $\mathcal{F}^{c}$ to get \eqref{eqn_theorem_Q_difference} holding with probability of at least $1 - \beta$ and this completes the proof.
\end{proof}

\subsection{Probabilistic $\mathcal{H}_{\infty}$ Performance Guarantees}
We will now connect all the above results with Problem \ref{problem_Hinf_performance} in the below theorem. 

\begin{theorem}
\label{theorem_probability_hinfty_performance_satisfaction}
Let the nominal controller $\bar{C}$ stabilise the nominal plant model $\bar{\Sigma}$ and achieves a certain $\mathcal{H}_{\infty}$ performance level denoted by $\bar{T} := \norm{\mathbf{T}(\bar{\Sigma}, \bar{C})}_{\mathcal{H}_{\infty}}$. Let $\tilde{\Sigma}(\theta) \in \mathcal{H}_{\infty}$ denote the perturbed plant such that $\mathrm{Gap}(\theta)$ is $\mathbf{L}_{\mathrm{gap}}$-Lipschitz by Theorem \ref{theorem_gap_lipschitz}. Given any desired performance level $\gamma > 0$ and a violation probability $\beta \in (0,1)$,  if
\begin{align}
\label{eqn_theorem_probability_hinfty_performance_satisfaction}
\frac{\gamma - \bar{T}}{1 + \gamma} - \mathbb{E}[\mathrm{Gap}(\theta)]
\geq 
\sqrt{2 \sigma^{2}_{\theta}
\mathbf{L}^{2}_{\mathrm{gap}}
\log\left(\frac{1}{\beta}\right)},
\end{align}
then, $\mathbb{P} \left( \norm{\mathbf{T}(\tilde{\Sigma}(\theta), \bar{C})}_{\mathcal{H}_{\infty}} \leq \gamma \right) \geq 1 - \beta$. 
\end{theorem}

\begin{proof}
From the robustness bounds in \eqref{eqn_performance_bounds_T}, we know that
\begin{align*}
\norm{\mathbf{T}(\tilde{\Sigma}(\theta), \bar{C})}_{\mathcal{H}_{\infty}} \leq 
\frac{\bar{T} + \mathrm{Gap}(\theta)}{1 - \mathrm{Gap}(\theta)}.
\end{align*}
For $\norm{\mathbf{T}(\tilde{\Sigma}(\theta), \bar{C})}_{\mathcal{H}_{\infty}} \leq \gamma$ to hold true, we need
\begin{align}
\label{eqn_thm5_interim_step_0}
\frac{\bar{T} + \mathrm{Gap}(\theta)}{1 - \mathrm{Gap}(\theta)}
\leq 
\gamma 
\iff
\mathrm{Gap}(\theta) 
\leq 
\frac{\gamma - \bar{T}}{1 + \gamma}.
\end{align}
Note that \eqref{eqn_thm5_interim_step_0} can be equivalently written as
\begin{align}
\label{eqn_bar_gamma_definition}
\mathrm{Gap}(\theta) - \mathbb{E}[\mathrm{Gap}(\theta)]
\leq 
\underbrace{\frac{\gamma - \bar{T}}{1 + \gamma} - \mathbb{E}[\mathrm{Gap}(\theta)]}_{
:= \bar{\gamma}}.   
\end{align}
Applying \eqref{eqn_corollary_prob_gap_exceed_threshold} from Corollary \ref{corollary_prob_gap_exceed_threshold} with $\epsilon = \bar{\gamma}$, we see that
\begin{align}
\label{eqn_thm5_interim_step_1}
\mathbb{P} \left( \mathrm{Gap}(\theta) - \mathbb{E}[\mathrm{Gap}(\theta)] > \bar{\gamma} \right) 
\leq 
\exp\left( -\frac{\bar{\gamma}^2}{2 \sigma^{2}_{\theta}
\mathbf{L}_{\mathrm{gap}}^2} \right).
\end{align}
For the above violation probability in \eqref{eqn_thm5_interim_step_1} to be upper bounded by $\beta$, we require
\begin{align}
\label{eqn_thm5_interim_step_2}
\exp\left( -\frac{\bar{\gamma}^2}{2 \sigma^{2}_{\theta} \mathbf{L}_{\mathrm{gap}}^2} \right)
\leq 
\beta
\iff
\eqref{eqn_theorem_probability_hinfty_performance_satisfaction}.
\end{align}
Hence, under this condition \eqref{eqn_thm5_interim_step_2}, $\norm{\mathbf{T}(\tilde{\Sigma}(\theta), \bar{C})}_{\mathcal{H}_{\infty}} \leq \gamma$ holds with probability at least $1 - \beta$.
\end{proof}

In the following corollary, we establish probabilistic guarantee on the desired $\mathcal{H}_{\infty}$ performance level $\gamma > 0$ being satisfied under the gap uncertainty.

\begin{corollary} \label{corollary_Hinf_performance_probability_lower_bound}
Under the conditions of Theorem \ref{theorem_probability_hinfty_performance_satisfaction}, given any desired performance level $\gamma > 0$,     
\begin{equation}
\label{eqn_Hinf_performance_probability_lower_bound}
\mathbb{P} \left( \norm{\mathbf{T}(\tilde{\Sigma}(\theta), \bar{C})}_{\mathcal{H}_{\infty}} \leq \gamma \right) 
\geq 
1 - \exp\left( -\frac{\bar{\gamma}^{2}}{2 \sigma^{2}_{\theta} \mathbf{L}_{\mathrm{gap}}^2} \right).
\end{equation}
\end{corollary}
\begin{proof}
Proof follows immediately by applying $\bar{\gamma}$ defined in  \eqref{eqn_bar_gamma_definition} into
\eqref{eqn_thm5_interim_step_1} and writing the complementary event. 
\end{proof}
\textbf{Remarks:} It is evident that the Lipschitz constant of the $\mathrm{Gap}(\theta)$ given by $\mathbf{L}_{\mathrm{gap}}$ from Theorem \ref{theorem_gap_lipschitz}, its expected value $\mathbb{E}[\mathrm{Gap}(\theta)]$ and the $\mathcal{H}_{\infty}$ performance level of the nominal controller with respect to the nominal model given by $\bar{T}$ dictates the probability of the perturbed model $\tilde{\Sigma}(\theta)$ achieving a certain $\mathcal{H}_{\infty}$ performance level given by $\gamma > 0$. \\

Having learned the probability of satisfying a given $\mathcal{H}_{\infty}$ performance level criteria in Theorem \ref{theorem_probability_hinfty_performance_satisfaction} and its corollary \ref{corollary_Hinf_performance_probability_lower_bound} for a lower bound on that probability, we now resort to find the expected $\mathcal{H}_{\infty}$ norm of the random transfer function $\mathbf{T}(\tilde{\Sigma}(\theta), \bar{C})$ due to the random gap perturbations. Before diving into the theorem that formally establishes the findings in connection to that, we first state and prove two lemmata in connection to get the expected value of certain quantities of interests which shall become handy while proving the required theorem.

\begin{lemma} \label{lemma_expectation_reciprocal_subgaussian}
Let $ x \in (0,1) $ be a real-valued random variable with mean $ \mu = \mathbb{E}[x] \in (0,1) $, and suppose that $x$ is sub-Gaussian with parameter $\sigma > 0$. Then, 
\begin{align}
\label{eqn_lemma_expectation_reciprocal_subgaussian}
\mathbb{E}\left[ \frac{1}{1 - x} \right]
\leq
\frac{1+\mu}{1-\mu}
+
\frac{8\sigma^2}{1 - \mu} e^{\left( - \frac{(1 - \mu)^2}{8\sigma^2} \right)}.
\end{align}
\end{lemma}

\begin{proof}
Note that the expectation of $x$ can be expressed using the tail integral identity. Further, we use the fact that $x$ is sub-Gaussian with mean $\mu$ to get x
\begin{align*}
\mathbb{E}\left[\frac{1}{1 - x}\right] 
&= 
\int_1^\infty \mathbb{P}\left( \frac{1}{1 - x} \geq t \right) dt \\
&= 
\int_1^\infty \mathbb{P}\left( x \geq 1 - \frac{1}{t} \right) dt. \\
&= 
\int_1^\infty
\mathbb{P}\left( x - \mu \geq 1 - \mu - \frac{1}{t} \right) dt \\
&\leq 
\int_1^\infty
e^{\left( -\frac{ \left(1 - \mu - \frac{1}{t}\right)^2}{2\sigma^2} \right)} dt.
\end{align*}
This is close to 1 when $\frac{1}{t} \rightarrow 1 - \mu$, and decays exponentially when $\frac{1}{t} \ll 1 - \mu$. So, we split the integral into, (i) a region close to the singularity, where the integrand could be large, (ii) a region far from the singularity, where the exponential decay dominates. Let $T := \frac{2}{1 - \mu}$ and we split the integral as 
\begin{align*}
\mathbb{E}\left[\frac{1}{1 - x}\right]
\leq \int_1^T 1 \, dt + \int_{T}^{\infty} e^{\left( -\frac{ \left(1 - \mu - \frac{1}{t}\right)^2}{2\sigma^2} \right)} dt.
\end{align*}
In the above inequality, we see that for $t \leq T$, $1 - \mu - \frac{1}{t}$ is small. Hence, we get a trivial yet conservative upper bound
\begin{align*}
\int^{T}_{1}
e^{\left( -\frac{ \left(1 - \mu - \frac{1}{t}\right)^2}{2\sigma^2} \right)} dt 
\leq 
\int^{T}_{1}
1
dt
=
T - 1
=
\frac{1+\mu}{1-\mu}.
\end{align*}
This is both a trivial and a conservative upper bound. On the other hand, $\forall t \geq T$, 
\begin{align*}
 &\frac{1}{t} \leq \frac{1 - \mu}{2}
 \\
 \iff 
 &1 - \mu - \frac{1}{t} \geq \frac{1 - \mu}{2} \\
 \iff
 &e^{\left( -\frac{ \left(1 - \mu - \frac{1}{t}\right)^2}{2\sigma^2} \right)}
\leq
e^{\left( -\frac{(1 - \mu)^2}{8\sigma^2} \right)}.
\end{align*}
\begin{align*}
\implies \mathbb{E}\left[\frac{1}{1 - x}\right]
\leq 
\frac{1+\mu}{1-\mu}
+ 
\int_{T}^{\infty} e^{\left( -\frac{(1 - \mu)^2}{8\sigma^2} \right)} dt.
\end{align*}
In order to have a finite upper bound for the second integral, we truncate the integral at length $ \frac{8\sigma^2}{1 - \mu} $ so that the tail contribution is negligent to get
\begin{align*}
\int_T^{T + \frac{8\sigma^2}{1 - \mu}} e^{\left( -\frac{(1 - \mu)^2}{8\sigma^2} \right)} dt = \frac{8\sigma^2}{1 - \mu} e^{\left( -\frac{(1 - \mu)^2}{8\sigma^2} \right)}.
\end{align*}
\begin{align*}
\implies \mathbb{E}\left[\frac{1}{1 - x}\right]
\leq
\frac{1+\mu}{1-\mu} +
\frac{8\sigma^2}{1 - \mu} \cdot e^{\left( - \frac{(1 - \mu)^2}{8\sigma^2} \right)},
\end{align*}
which completes the proof.
\end{proof}

Using the above lemma \ref{lemma_expectation_reciprocal_subgaussian}, we can find the expectation of $1/(1 - \mathrm{Gap}(\theta))$ which will be useful for us later which obtaining bounds for the expected value of $\norm{\mathbf{T}(\tilde{\Sigma}(\theta), \bar{C})}_{\mathcal{H}_{\infty}}$.

\begin{lemma}\label{lemma_expected_inverse_gap}
Given that $\mathrm{Gap}(\theta)$ is a sub-Gaussian random variable in $(0,1)$ with parameter $\sigma_{\theta} \mathbf{L}_{\mathrm{gap}}$, we see that    
\begin{align}
\label{eqn_lemma_expected_inverse_gap}
\mathbb{E}\left[ \frac{1}{1 - \mathrm{Gap}(\theta)} \right]
\leq
\underbrace{
\frac{1+C_{\mathrm{gap}}}{1-C_{\mathrm{gap}}} +
\frac{8 \sigma^{2}_{\theta} \mathbf{L}_{\mathrm{gap}}^{2}}{1 - C_{\mathrm{gap}}} e^{\left( - \frac{(1 - C_{\mathrm{gap}})^2}{8 \sigma^{2}_{\theta} \mathbf{L}_{\mathrm{gap}}^{2}} \right)}}_{:= C^{\mathrm{inv}}_{\mathrm{gap}}}.
\end{align}
\end{lemma}
\begin{proof}
Since $\mathrm{Gap}(\theta)$ is a sub-Gaussian random variable in $(0,1)$ with parameter $\sigma_{\theta} \mathbf{L}_{\mathrm{gap}}$, we employ Lemma \ref{lemma_expectation_reciprocal_subgaussian} with $x = \mathrm{Gap}(\theta)$ and $\mu = \mathbb{E}[\mathrm{Gap}(\theta)]$ in \eqref{eqn_lemma_expectation_reciprocal_subgaussian} to see that 
\begin{align}
\label{eqn_lemma_expected_inverse_gap_interim_1}
&\mathbb{E}\left[ \frac{1}{1 - \mathrm{Gap}(\theta)} \right] \nonumber \\
&\leq
\frac{1+\mathbb{E}[\mathrm{Gap}(\theta)]}{1-\mathbb{E}[\mathrm{Gap}(\theta)]} +
\frac{8 \sigma^{2}_{\theta} \mathbf{L}_{\mathrm{gap}}^{2}}{1 - \mathbb{E}[\mathrm{Gap}(\theta)]} e^{\left( - \frac{(1 - \mathbb{E}[\mathrm{Gap}(\theta)])^2}{8 \sigma^{2}_{\theta} \mathbf{L}_{\mathrm{gap}}^{2}} \right)}.
\end{align}
Notice from \eqref{eqn_corollary_lemma_expected_gap} that 
\begin{align*}
&\mathbb{E}[\mathrm{Gap}(\theta)]
\leq 
C_{\mathrm{gap}} \\
\iff
&1 + \mathbb{E}[\mathrm{Gap}(\theta)]
\leq 
1 + C_{\mathrm{gap}} \\
\iff
&1-\mathbb{E}[\mathrm{Gap}(\theta)]
\geq 
1-C_{\mathrm{gap}} \\
\iff
&\frac{1}{1-\mathbb{E}[\mathrm{Gap}(\theta)]}
\leq
\frac{1}{1-C_{\mathrm{gap}}}
\end{align*}
Similarly, 
\begin{align*}
&1-\mathbb{E}[\mathrm{Gap}(\theta)]
\geq 
1-C_{\mathrm{gap}} \\
\iff
&(1-\mathbb{E}[\mathrm{Gap}(\theta)])^{2}
\geq 
(1-C_{\mathrm{gap}})^{2} \\
\iff
&-(1-\mathbb{E}[\mathrm{Gap}(\theta)])^{2}
\leq 
-(1-C_{\mathrm{gap}})^{2} \\
\iff
&e^{\left( - \frac{(1 - \mathbb{E}[\mathrm{Gap}(\theta)])^2}{8 \sigma^{2}_{\theta} \mathbf{L}_{\mathrm{gap}}^{2}} \right)}
\leq
e^{\left( - \frac{(1 - C_{\mathrm{gap}})^2}{8 \sigma^{2}_{\theta} \mathbf{L}_{\mathrm{gap}}^{2}} \right)}.
\end{align*}
Using all these facts from \eqref{eqn_corollary_lemma_expected_gap} in \eqref{eqn_lemma_expected_inverse_gap_interim_1}, we get \eqref{eqn_lemma_expected_inverse_gap}. 
\end{proof}

Having stated and proved the necessary lemmata, we now turn our attention to state and prove the following theorem on finding a bound for the expected $\mathcal{H}_{\infty}$ norm of the transfer function $\mathbf{T}(\tilde{\Sigma}(\theta), \bar{C})$.

\begin{theorem}\label{theorem_expected_Hinf_norm_Tzw}
Let $\bar{\Sigma}$ be a nominal LTI system model with a stabilising nominal controller $\bar{C}$, and let $\tilde{\Sigma}(\theta)$ be a random plant stabilised by $\bar{C}$ with $\theta \sim \mathcal{N}(\mu_{\theta}, \sigma^{2}_{\theta} I_{p})$ such that \eqref{eqn_gap_Lipschitz} holds. Let $\bar{b} := b_{\bar{\Sigma}, \bar{C}}, \tilde{T}(\theta) = \norm{\mathbf{T}(\tilde{\Sigma}(\theta), \bar{C})}_{\mathcal{H}_{\infty}}$. If 
\begin{align}
\label{eqn_condition_theorem_expected_Hinf_norm_Tzw}
\norm{\mathbf{T}(\bar{\Sigma}, \bar{C})}_{\mathcal{H}_{\infty}} \leq b_{\bar{\Sigma}, \bar{C}},
\end{align}
and $\mathrm{Gap}(\theta) < 1$,  then
\begin{align}
\label{eqn_theorem_expected_Hinf_norm_Tzw}
\mathbb{E} 
\left[ \tilde{T}(\theta) \right] 
< 
(\bar{b} + 1) \, 
C^{\mathrm{inv}}_{\mathrm{gap}}.
\end{align}
\end{theorem}

\begin{proof}
Since $\mathrm{Gap}(\theta) < 1$ in this setting, we employ \eqref{eqn_performance_degradation} and \eqref{eqn_condition_theorem_expected_Hinf_norm_Tzw} on \eqref{eqn_performance_bounds_T} to get
\begin{align}
\label{eqn_theorem_expected_Hinf_norm_Tzw_interim_step_1}
\tilde{T}(\theta)
\leq
\frac{\norm{\mathbf{T}(\bar{\Sigma}, \bar{C})}_{\mathcal{H}_{\infty}} + \mathrm{Gap}(\theta)}{1 - \mathrm{Gap}(\theta)}
< \frac{\bar{b} + 1}{1 - \mathrm{Gap}(\theta)}.
\end{align}
Taking expectations on both sides and using linearity of expectation, we get
\begin{align}
\label{eqn_theorem_expected_Hinf_norm_Tzw_interim_step_2}
\mathbb{E}[\tilde{T}(\theta)] 
< 
(\bar{b} + 1) \, 
\mathbb{E}
\left[
\frac{1}{1 - \mathrm{Gap}(\theta)}
\right].
\end{align}
Using \eqref{eqn_lemma_expected_inverse_gap} in \eqref{eqn_theorem_expected_Hinf_norm_Tzw_interim_step_2}, we get
\eqref{eqn_theorem_expected_Hinf_norm_Tzw}.
\end{proof}

\subsection{Numerical Simulation} 
Consider the following nominal SISO system model $\bar{\Sigma}: (A, B, C, D) = (-1, 1, 1, 0)$ which corresponds to the following state space form:
\begin{align}
\label{eqn_simulation_nominal_model}
\bar{\Sigma}
:
\Bigl\{
\dot{x} = -x + u, \quad y = x.
\end{align}
We obtained the nominal controller $\bar{C}$ by placing the poles of $\bar{\Sigma}$ at $-2$. The value of performance measure was found to be $b_{\bar{\Sigma}, \bar{C}} = 0.7071$ and $\gamma$ values were varied between $[1.01, 3]$. We obtained $N = 10^{4}$ samples of $\theta \sim \mathcal{N} \left( \begin{bmatrix}
0.1 \\ -0.05 \end{bmatrix}, \sigma^{2}_{\theta} I_{2} \right)$, with $\sigma_{\theta} = 0.5$. Using the samples $\left\{ \theta^{(i)} \right\}^{N}_{i=1}$, we constructed $N$ different perturbed plant models $\left\{\tilde{\Sigma}(\theta^{(i)}) := (A+\theta^{(i)}_{1}, B, C+\theta^{(i)}_{2}, D) \right\}^{N}_{i=1}$. The gap between each of the perturbed $\tilde{\Sigma}(\theta^{(i)})$ and nominal model $\bar{\Sigma}$ was computed using the \emph{\textrm{gapmetric}} command of the Matlab. The result is shown in Figure \ref{fig:probHinfGuarantee}. \\

Doing a Monte-carlo style simulation by generating $N = 10^{4}$ different independent instances of the perturbed plant models $\tilde{\Sigma}(\theta)$, we estimated the Lipschitz constant of the perturbed model $\tilde{\Sigma}(\theta)$ with respect to $\theta$ as $\mathbf{L}_{\mathrm{gap}} = 0.5308$. Further, we estimated $\mathbb{E}\left[ \mathrm{Gap}(\theta) \right] = 0.3204$ and its upper bound $C_{\mathrm{gap}} = 0.4023$ using \eqref{eqn_corollary_lemma_expected_gap}. Further, the $\mathbb{E} \left[ \norm{T_{zw}(\tilde{\Sigma}(\theta), \bar{C})}_{\mathcal{H}_{\infty}} \right]$ was estimated to be $0.5557$ and its conservative upper bound from \eqref{eqn_theorem_expected_Hinf_norm_Tzw} in Theorem \ref{theorem_expected_Hinf_norm_Tzw} was $4.8592$. As noted previously, this is a conservative estimate given the simple upper bound used in Theorem \ref{theorem_expected_Hinf_norm_Tzw}.

\begin{figure}
    \centering
    \includegraphics[width=\linewidth]{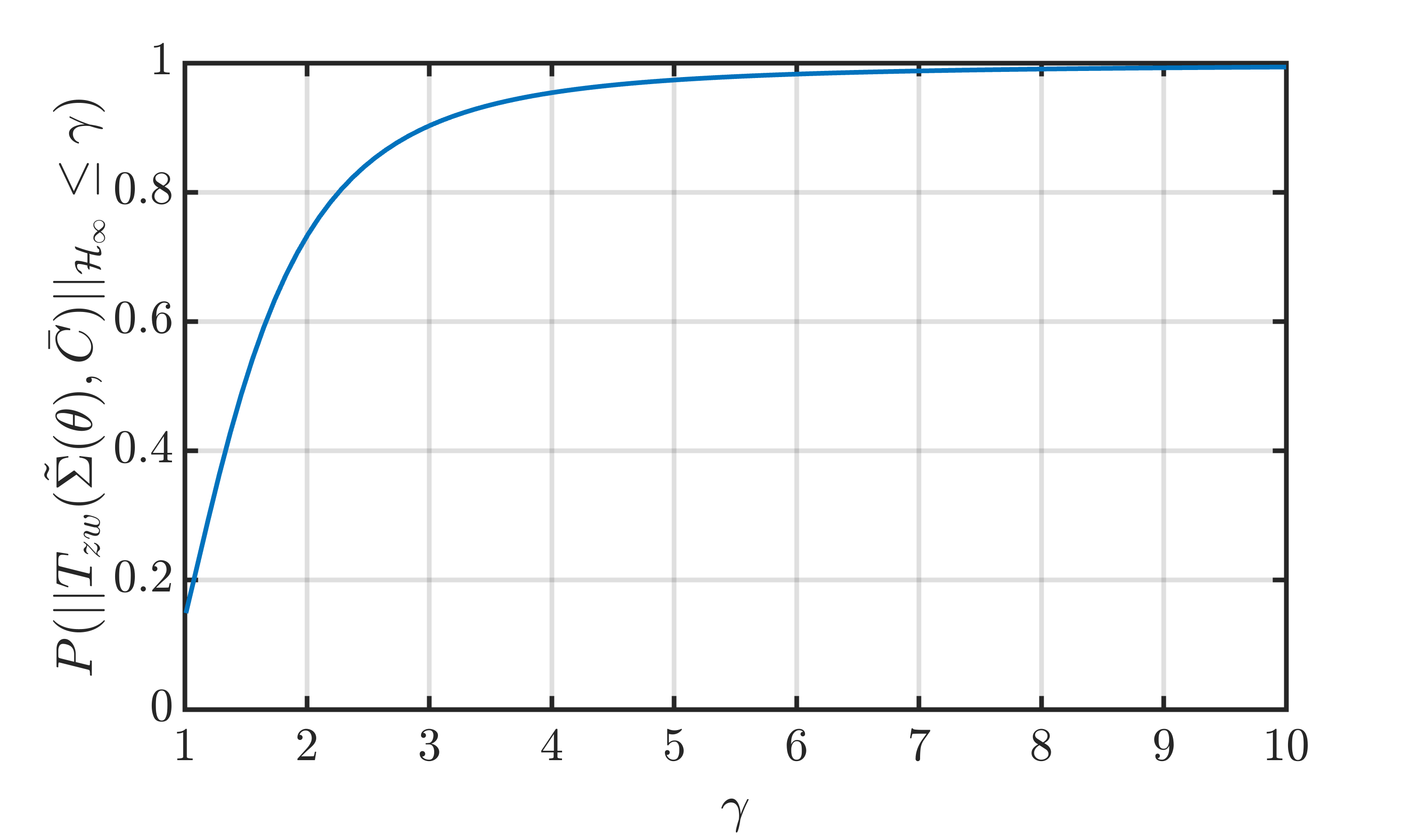}
    \caption{The lower bound on the $\mathbb{P} \left( \norm{\mathbf{T}(\tilde{\Sigma}(\theta), \bar{C})}_{\mathcal{H}_{\infty}} \leq \gamma \right)$ for different values of $\gamma$ using Theorem \ref{theorem_expected_Hinf_norm_Tzw} are plotted here. }
    \label{fig:probHinfGuarantee}
\end{figure}


\section{Connections to Existing Probabilistic Robust Control (PRC) Theory}
The analysis of stochastic robustness of LTI systems started in \cite{stengel1991technical} where authors studied a very similar problem as in \eqref{eqn_perturbed_dynamics_system} and gave estimates of stability probability density functions for systems affected by uncertain parameters. On the other hand, PRC approaches gained traction further later on using probabilistic methods to give guarantees on system stability and performance (see \cite{khargonekar1996randomized, tempo2013randomized, calafiore2007probabilistic} and the references therein for further details on this topic.). Having said that one might be interested on seeing what is new with the approach proposed in this paper regarding PRC and what new perspectives does this bring to the already existing table of approaches for PRC theory. Existing PRC theories utilise the stability margins to gauge the stability of the controller-plant pair, while the approach considered in this paper ties everything like stability and performance nicely with the $\mathrm{Gap}(\theta)$ and gives probabilistic guarantees. 

\subsection{Connecting Random Gap \& Scenario-Based Robustness}
In scenario-based approaches, one works with just the finite number of samples of uncertainties possibly from an unknown generating distribution to do both the reliability estimation and performance estimation in the context of PRC setting considered in \cite{calafiore2007probabilistic}. To connect the scenario-based approach with our random gap based problem formulation, we will deviate from the assumption that the distribution of the uncertain parameter $\theta$ is known in this section meaning that $\mathbf{f}_{\theta}$ is not necessarily equal to $\mathcal{N}(\mu_{\theta}, \sigma^{2}_{\theta} I_{p})$. We have the following assumption in place in regards to that. 
\begin{assumption}
\label{assume_samples}
The distribution $\mathbf{f}_{\theta}$ is unknown but we have access to a finite set of $N \in \mathbb{N}$ independent and identically distributed samples $\{\theta^{(i)}\}_{i=1}^{N}$ drawn from $\mathbf{f}_{\theta}$.    
\end{assumption}

Assumption \ref{assume_samples} just says that we have $N$ samples of $\theta$ available for decision making. We know that if $N \rightarrow \infty$, it means that we essentially know the distribution $\mathbf{f}_{\theta}$ exactly and thereby getting rid of the uncertainty associated with distribution of $\theta$ in assumption \ref{assume_samples}. On the other hand, if we were to give probabilistic guarantees on the nominal controller stabilising a random plant based on just the available $N$ samples of $\theta$, then the resulting probability will be determined by $N$. Usually, a failure or a violation probability is given apriori and we need to find a connection between that violation probability and the number of samples $N$ to give probabilistic guarantees. The following theorem nicely establishes a connection by leveraging the power of scenario-based approaches to give finite sample-based guarantees on the gap-metric based robust stability.

\begin{theorem}
\label{theorem_scenario_gap_robustness}
Let $\bar{C}$ be a nominal controller stabilising a nominal plant $\bar{\Sigma}$ and results in $b_{\bar{\Sigma}, \bar{C}} > 0$.
Suppose that the perturbed plant $\tilde{\Sigma}(\theta) \in \mathcal{H}_{\infty}$ be affected by a random parameter $\theta \in \mathbb{R}^{p}$ characterised by assumption \ref{assume_samples}. Given a confidence level $\beta\in(0,1)$ and a violation probability $\epsilon\in(0,1)$, if the number of samples of $\theta$ satisfies
\begin{align}
\label{eqn_samplesize_condition}
N \geq
\frac{\log\left(\frac{1}{\beta}\right)}{\log\left(\frac{1}{1-\epsilon}\right)}, \quad \text{then},
\end{align}
\begin{align}
\label{eqn_probust_scenario_based_gap}
\mathbb{P}\left(
\mathbb{P}\left(
\mathrm{Gap}(\theta)
\leq 
\max_{i = 1,\dots,N}\mathrm{Gap}(\theta^{(i)}) 
\right)
\geq 1-\epsilon 
\right)
\geq 1-\beta.
\end{align}
Additionally, if $\max_{i = 1,\dots,N}\mathrm{Gap}(\theta^{(i)}) < b_{\bar{\Sigma}, \bar{C}}$, then
\begin{align}
\label{eqn_prob_nominal_c_stab_perturbed_plant}
\mathbb{P}( \bar{C} \text{ stabilizes } \tilde{\Sigma}(\theta))
\geq 
1 - \epsilon.
\end{align}
\end{theorem}

\begin{proof}
For a violation threshold $\alpha > 0$, let us denote the probability of violation as
\begin{align}
\label{eqn_gap_violation_event}
V(\alpha)
=
\mathbb{P}(\mathrm{Gap}(\theta) > \alpha).
\end{align}
Further, let
\begin{align}
\label{eqn_max_finite_gaps}
\hat{\alpha}_{N}
:= 
\max_{i = 1,\dots,N}\mathrm{Gap}(\theta^{(i)}).
\end{align}
Then, from Theorem 1 in \cite{calafiore2007probabilistic}, we know that for a given violation probability $\epsilon \in (0,1)$, if $N$ satisfies \eqref{eqn_samplesize_condition}, then
\begin{align}
\label{eqn_violation_prob_confidence}
\mathbb{P}\left(
V(\hat{\alpha}_{N})
\leq
\epsilon
\right) 
\geq 
1 - \beta.
\end{align}
Additionally, if $\hat{\alpha}_{N}$ satisfies $\hat{\alpha}_{N} < b_{\bar{\Sigma}, \bar{C}}$, then from \eqref{eqn_b_pc_definition}, we see that every plant within a gap radius $\hat{\alpha}_{N}$ from $\bar{\Sigma}$ is stabilized by $\bar{C}$. When $\hat{\alpha}_{N} < b_{\bar{\Sigma}, \bar{C}}$, from \eqref{eqn_gap_violation_event} and \eqref{eqn_violation_prob_confidence}, we know that $\mathrm{Gap}(\theta) \leq \hat{\alpha}_{N} < b_{\bar{\Sigma}, \bar{C}}$, would imply that $\bar{C}$ will stabilize $\tilde{\Sigma}(\theta))$. We know that if event $B$ implies event $A$, then $\mathbb{P}(B) \leq \mathbb{P}(A)$. Hence,
\begin{align*}
\mathbb{P}( \bar{C} \text{ stabilizes } \tilde{\Sigma}(\theta))
&\geq
\mathbb{P}(
\mathrm{Gap}(\theta)
\leq 
\hat{\alpha}_{N}
) \\ 
&=
1-V(\hat{\alpha}_{N}) \\
&\geq 
1-\epsilon.
\end{align*}
This completes the proof.
\end{proof}

\subsection{Interpretation of Confidence and Probabilistic Robustness}
The probabilistic robustness guarantee obtained through scenario-based method in \eqref{eqn_probust_scenario_based_gap} of Theorem \ref{theorem_scenario_gap_robustness} involves two distinct sources of randomness namely,
\begin{enumerate}
    \item \emph{Randomness from $\mathbf{f}_{\theta}$ (Probabilistic Robustness):} The term $\mathbb{P}\left(\mathrm{Gap}(\theta)>\hat{\alpha}_{N}\right)\leq\epsilon$ describes the probability (under the unknown $\mathbf{f}_{\theta}$ of $\theta$) that a randomly chosen uncertain plant model $\tilde{\Sigma}(\theta)$ results in the gap exceeding threshold $\hat{\alpha}_{N}$.

    \item \emph{Randomness from Scenario Sampling (Confidence):} 
    Since the scenarios $\{\theta^{(i)}\}_{i=1}^{N}$ are drawn randomly from the unknown $\mathbf{f}_{\theta}$, the scenario-based gap threshold $\hat{\alpha}_{N}$ given by \eqref{eqn_max_finite_gaps} itself is random. Hence, the event $\mathbb{P}\left(\mathrm{Gap}(\theta) > \hat{\alpha}_{N} \right) \leq \epsilon$ is also random. The confidence level $1-\beta$ quantifies the probability (over repeated scenario samplings) that the scenario-based threshold correctly achieves the probabilistic robustness guarantee. 
\end{enumerate}
Thus, the probabilistic robustness guarantee is itself a random quantity due to scenario sampling, and the confidence level quantifies our trust to obtain a good scenario-based threshold.

\subsection{Numerical Simulation} 
In order to demonstrate the connection of the random gap with the scenario-based robustness established in Theorem \ref{theorem_scenario_gap_robustness}, we considered a nominal system $\bar{\Sigma}: \frac{1}{s+1}$ and a perturbed system $\tilde{\Sigma}(\theta) = \frac{1}{s+(1+\theta)}$, where $10^{4}$ samples of $\theta$ were sampled from $\mathcal{N}(0, {0.25}^{2})$. We chose a violation probability of $\epsilon = 0.05$ and a confidence level parameter of $\beta = 0.01$. For the nominal controller $\bar{C} = \frac{(s+2)}{2(s+3)}$, the performance measure was $b_{\bar{\Sigma}, \bar{C}} = 0.8944$. Given a violation probability of $\epsilon = 0.05$ and a confidence level of $\beta = 0.01$, the sample size condition from \eqref{eqn_samplesize_condition} resulted in $N \geq 90$. By running a Monte-carlo simulation using $10^{4}$ instances of $\theta$ generated as mentioned above, the results are plotted in Figure \ref{fig:scenario_gap}. 

\begin{figure}
    \centering
    \includegraphics[width=\linewidth]{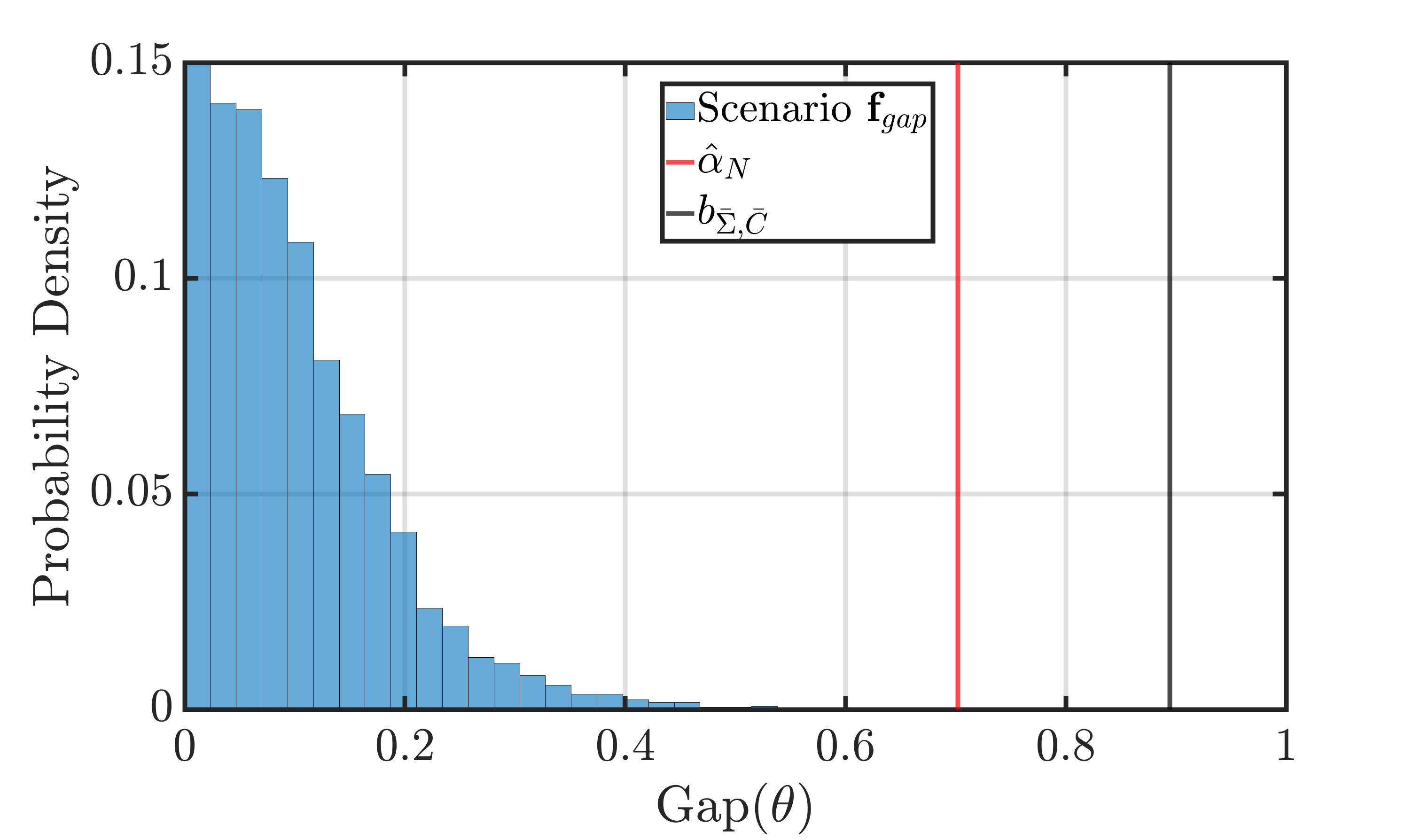}
    \caption{The probability density of $\mathrm{Gap}(\theta)$ is plotted for $10^{4}$ independent samples of $\theta$. Since \eqref{eqn_samplesize_condition} holds and $\hat{\alpha}_{N} < b_{\bar{\Sigma}, \bar{C}}$, the nominal controller $\bar{C}$ would end up stabilising the random plant with probability of at least $1-\epsilon$ as per Theorem \ref{theorem_scenario_gap_robustness}.}
    \label{fig:scenario_gap}
\end{figure}

We observed that $\hat{\alpha}_{N} = 0.7019$ which was less than $b_{\bar{\Sigma}, \bar{C}} = 0.8944$. This ensured that from Theorem \ref{theorem_scenario_gap_robustness} that the condition $\mathrm{Gap}(\theta) \leq \hat{\alpha}_{N}$ holding with probability of at least $1 - \epsilon = 0.95$ implied that the nominal controller $\bar{C} $ would stabilise the random plant $\tilde{\Sigma}(\theta)$ with probability of at least the same value of $1 - \epsilon = 0.95$. As mentioned earlier, the confidence level $1-\beta$ quantifies the probability (over repeated scenario samplings) that the scenario-based threshold correctly achieves the probabilistic robustness guarantee. 
That is, if we repeatedly draw new sets of $N \geq 90$ samples, then in at least $100(1-\beta)\%$ of these repetitions, the computed scenario threshold $\hat{\alpha}_{N}$ will ensure a true violation probability no greater than $\epsilon$. On the other hand, when the covariance strength of $\theta$ was increased to ${0.5}^{2}$ from ${0.25}^{2}$, we were able to see some unstable perturbed plants being generated meaning that $\mathrm{Gap}(\theta)$ and hence $\hat{\alpha}_{N}$ became equal to $1$ for those samples of $\theta$. Note that even under such unstable case, \eqref{eqn_probust_scenario_based_gap} would hold true but unfortunately \eqref{eqn_prob_nominal_c_stab_perturbed_plant} will not hold as $\underbrace{\max_{i = 1,\dots,N}\mathrm{Gap}(\theta^{(i)})}_{ = 1} >  \underbrace{b_{\bar{\Sigma}, \bar{C}}}_{ = 0.8944}$.


\section{Conclusion \& Future Outlook} \label{sec_conclusion}

When a random parameter affects a linear system, we studied how it manifested itself as the associated random gap between the nominal model of the system (without any uncertainty) and the perturbed model of the system (with uncertainty). The randomness in the associated gap resulted in probabilistic versions of the corresponding performance guarantees and stability margins guarantees measured in terms of the gap. This new perspective on PRC using the random gap provides us information about upper bounds on the expected gap quantity and the expected $\mathcal{H}_{\infty}$ achievable performance level apriori for any stabilising controller. A connection to the existing tools on PRC using scenario-based approach was also presented in this paper. The results obtained in this paper nicely blends the high dimensional statistics tool with the random gap problem formulation and gives probabilistic guarantees on both gap metric based robust performance and robust stability. \\

The aim of this paper is to revive the research on PRC theory. Future research prospects look very promising and there are many interesting open research questions along the lines of the research presented in this paper. We list here few of them which we believe can be immediately pursued given the existing developments done in this manuscript. 
\begin{enumerate}
    \item Obtain probabilistic guarantees on gap metric by investigating the randomness in the projection operator. 
    \item Investigate and give bounds on the expected distance between $\delta_{g}(\bar{C}, C(\theta))$, where $C(\theta)$ would be the controller which will result in same performance measure for the perturbed system $\tilde{\Sigma}(\theta)$ as $\bar{C}$ did for the nominal system $\bar{\Sigma}$ meaning that $b_{\tilde{\Sigma}(\theta), C(\theta)} = b_{\bar{\Sigma}, \bar{C}}$. This would inform us how much the respective controllers that guarantee same performance level for the nominal and perturbed plants are further apart in the expected sense. 
    \item An important future research direction would be to formulate and compute the distance between two stochastic dynamical systems $\delta_{\mathrm{gap}}(\tilde{\Sigma}_{1}(\theta), \tilde{\Sigma}_{2}(\theta))$ where the uncertainty in each system is described using the random gap metric between its respective nominal model and the perturbed model.
    \item Another important research direction will be to extend the problem setting to both linear time varying systems and to nonlinear systems by formulating the quantity of interest namely the gap between the nominal and the corresponding perturbed system models as random. 
    \item Another interesting direction is to first develop gap metric based robust tube model predictive control (MPC). The uncertainty around the system trajectories from the true but unknown perturbed model different from the nominal model is characterised along the prediction horizon using the assumed gap between the nominal $(\bar{P})$ and the perturbed system $(P)$ (by formulating linear matrix inequality (LMI) \cite{scherer2000linear} constraints for the condition $\delta_{g}(P, \bar{P}) \leq \alpha$ for a given $\alpha \in (0, 1)$). Using the random gap based problem formulation considered in this manuscript, the deterministic gap metric based robust tube MPC can be even extended further to gap metric based stochastic tube MPC setting.
\end{enumerate}


\section*{Acknowledgment}
The author would like to thank 
Dr. Anders Rantzer from Lund University, Sweden for pointing out to the valuable resource materials while the author started his research on probabilistic robust control. The author is grateful to Dr. Michael Cantoni from the University of Melbourne, Australia for pointing out the mistakes in the problem formulation and for thoughtful guidance. The author is also grateful to Dr. Tryphon Georgiou from the University of California, Irvine, USA for explaining him with the nuances of the gap metric while attending the ECC 2024 held at Stockholm, Sweden.


\bibliographystyle{IEEEtran}
\bibliography{references}

\vskip 0pt plus -1fil


\begin{IEEEbiography}[{\includegraphics[width=1in,height=1.25in,clip,keepaspectratio]{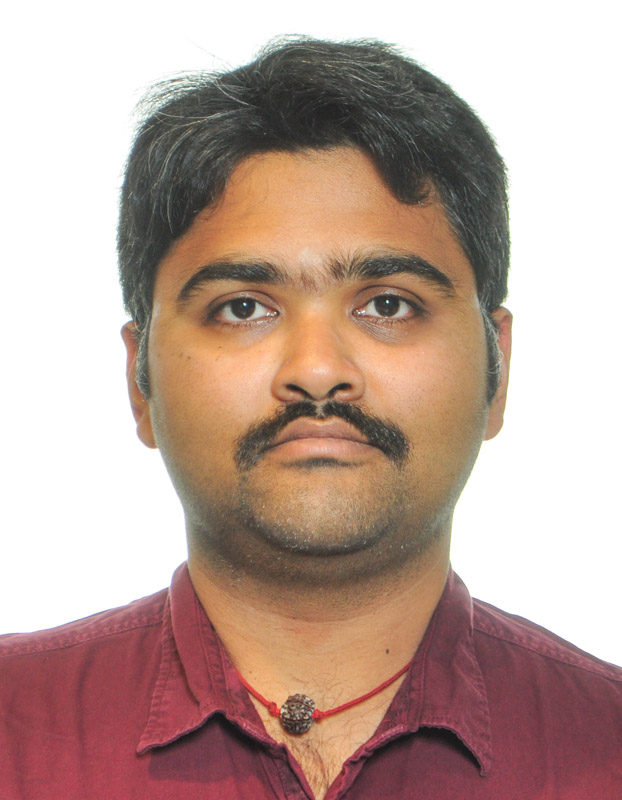}}]{Venkatraman Renganathan} (Member, IEEE) received his Bachelors degree in electrical and electronics engineering from the Government College of Technology, Anna University, Coimbatore, India in 2011. Further, he obtained his Masters degree in electrical engineering with focus on control systems from Arizona State University, USA, in 2016 and the doctoral degree in mechanical engineering with emphasis on dynamics and control from The University of Texas at Dallas, USA in 2021. He was a postdoctoral fellow at the department of automatic control at Lund University in Sweden from August 2021 till August 2024. Currently, he is a lecturer at the faculty of engineering \& applied sciences in Cranfield University, United Kingdom. His research interests include probabilistic robust control, distributed adaptive control for uncertain networks, and risk bounded motion planning.
\end{IEEEbiography}

\end{document}

%% file: commands.tex
\def\bbn{\mathbb N}
\def\bbz{\mathbb Z}
\def\bbr{\mathbb R}

\newcommand{\norm}[1]{\left\lVert {#1} \right\rVert}


%% file: robustcontroldiagram.tex
\tikzstyle{block} = [draw, fill=orange!20, rectangle, 
    minimum height=3em, minimum width=6em]
    \begin{tikzpicture}[auto, node distance=1cm, >=latex']
    
        \node [draw, rectangle, fill=orange!20, minimum width=2cm, minimum height=1cm] (G) {$\bar{\Sigma}$};
        \node [draw, rectangle, minimum width=2cm, fill=orange!20, minimum height=1cm, below=of G] (K) {$K$};
        \node [draw, rectangle, minimum width=2cm, fill=orange!20, minimum height=1cm, above=of G] (D) {$\Delta$};
    
        \node [left=of G.180] (w_in) {$w$};
        \node [right=of G.3] (y_out) {$z$};
    
        \draw [->] (w_in) -- (G.180);
        \draw [->] (G.3) -- (y_out);
        
        \draw [->] (G.340) -- node[below] {$y$} ++(1,0) |- (K);
        \draw [->] (K) -| ++(-2,0) |- node[left] {$u$} (G.200);
    
        \draw [->] (G.20) -- node[below] {} ++(1,0) |- (D);
        \draw [->] (D) -| ++(-2,0) |- node[left] {} (G.160);
    
        \draw[fill = blue, nearly transparent] ($(-2.5,-0.75) + (G) $) rectangle ($(2.6, 1.25) + (D) $);
        \node[align=right] at  ($(0.0, 1.0) + (D) $) {\textbf{Perturbed Plant} $(\tilde{\Sigma})$};
    
    \end{tikzpicture}